\pdfoutput=1
\documentclass[12pt]{article}
\usepackage[margin=1in]{geometry}

\usepackage[backend=biber,style=apa,uniquename=false,uniquelist=false,sorting=nyt]{biblatex}
\renewcommand{\cite}[1]{\parencite{#1}}

\usepackage[T1]{fontenc}
\usepackage[english]{babel}
\usepackage{csquotes}

\usepackage{url}
\PassOptionsToPackage{hyphens}{url}

\usepackage{hyperref}
\PassOptionsToPackage{breaklinks}{hyperref}

\usepackage{xcolor}
\hypersetup{colorlinks, linkcolor={blue!50!black}, citecolor={blue!50!black}, urlcolor={blue!80!black}}

\usepackage{tikz}
\usetikzlibrary{external} 
\usepackage{pgfplots}
\pgfplotsset{compat=1.12}

\newcommand{\paragraphNoSkip}[1]{\paragraph{#1.}}

\usepackage{mathtools}
\usepackage{amssymb}
\usepackage{amsthm}
\usepackage{thmtools,thm-restate}
\usepackage{enumitem}

\newtheorem{definition}{Definition}

\newtheorem{example}{Example}

\usepackage[capitalise]{cleveref}
\crefname{definition}{Definition}{Definitions}
\crefname{example}{Example}{Examples}

\usepackage[symbols,acronym,nonumberlist,nogroupskip,stylemods={mcols,longbooktabs},section=subsection,numberedsection]{glossaries-extra}
\makenoidxglossaries
\glssetcategoryattribute{acronym}{nohyper}{true}
\setabbreviationstyle[acronym]{long-short}

\newacronym{iff}{iff}{if and only if}
\newacronym{iid}{iid}{independent and identically distributed}
\newacronym{wlog}{wlog}{without loss of generality}
\newacronym{DeFi}{DeFi}{decentralized finance}
\newacronym{GE}{GE}{generalized entropy}
\newacronym{GC}{GC}{Gini coefficient}
\newacronym{CV}{CV}{coefficient of variation}
\newacronym{TL}{TL}{Theil L}
\newacronym{TT}{TT}{Theil T}
\newacronym{HH}{HH}{Herfindahl-Hirschman}
\newacronym{AM}{AM}{Atkinson measure}
\newacronym{BBG}{BBG}{bigger is better Gini coefficient}

\newcommand{\define}{\stackrel{\mathclap{\mbox{\text{\scriptsize def}}}}{=}}

\glsxtrnewsymbol[description={The number of ``true'' unique actors inhabiting the economy.}]{wealthn}{\ensuremath{n}}
\newcommand{\wealthn}{{\gls[hyper=false]{wealthn}}}

\glsxtrnewsymbol[description={The set of ``true'' unique actors inhabiting the economy.}]{wealths}{\ensuremath{\left[\wealthn\right]}}
\newcommand{\wealths}{{\gls[hyper=false]{wealths}}}

\glsxtrnewsymbol[description={Hidden wealth.}]{wealth}{\ensuremath{w}}
\newcommand{\wealth}{{\gls[hyper=false]{wealth}}}

\glsxtrnewsymbol[description={The number of observable pseudonymous identities.}]{owealthn}{\ensuremath{m}}
\newcommand{\owealthn}{{\gls[hyper=false]{owealthn}}}

\glsxtrnewsymbol[description={The set of observable pseudonymous identities.}]{owealths}{\ensuremath{\left[\owealthn\right]}}
\newcommand{\owealths}{{\gls[hyper=false]{owealths}}}

\glsxtrnewsymbol[description={Observable wealth.}]{owealth}{\ensuremath{\rho}}
\newcommand{\owealth}{{\gls[hyper=false]{owealth}}}

\glsxtrnewsymbol[description={The reported wealth of each real actor across observable identities.}]{report}{\ensuremath{R}}
\newcommand{\report}{{\gls[hyper=false]{report}}}

\glsxtrnewsymbol[description={A coalition of real actors.}]{coalition}{\ensuremath{C}}
\newcommand{\coalition}{{\gls[hyper=false]{coalition}}}

\glsxtrnewsymbol[description={An inequality measure.}]{inequality}{\ensuremath{\mathcal{I}}}
\newcommand{\inequality}{{\gls[hyper=false]{inequality}}}

\glsxtrnewsymbol[description={An inequality measure which is constant for any input.}]{constant}{\ensuremath{\inequality_c}}
\newcommand{\constant}{{\gls[hyper=false]{constant}}}

\glsxtrnewsymbol[description={The constant $0$ function.}]{baseline}{\ensuremath{\inequality_0}}
\newcommand{\baseline}{{\gls[hyper=false]{baseline}}}

\glsxtrnewsymbol[description={An aggregation matrix.}]{agg}{\ensuremath{A}}
\newcommand{\agg}{{\gls[hyper=false]{agg}}}

\glsxtrnewsymbol[description={The \gls{GC}.}]{gc}{\ensuremath{\inequality_{\glsxtrshort{GC}}}}
\newcommand{\gc}{{\gls[hyper=false]{gc}}}

\glsxtrnewsymbol[description={The \gls{HH} index.}]{hh}{\ensuremath{\inequality_{\glsxtrshort{HH}}}}

\begin{document}
\title{Inequality in the Age of Pseudonymity}
\author{
    Aviv Yaish\footnote{Email: \href{mailto:a@yai.sh}{a@yai.sh}} \\
    \small{Yale University \& IC3}
    \and
    Nir Chemaya\footnote{Email: \href{mailto:chemayan@bgu.ac.il}{chemayan@bgu.ac.il}} \\
    \small{Ben-Gurion University}
    \and
    Lin William Cong\footnote{Email: \href{mailto:will.cong@cornell.edu}{will.cong@cornell.edu}} \\
    \small{Cornell University ABFER,} \\
    \small{IC3, \& NBER}
    \and
    Dahlia Malkhi\footnote{Email: \href{mailto:dahliamalkhi@ucsb.edu}{dahliamalkhi@ucsb.edu}} \\
    \small{UC Santa Barbara}
}
\date{}
\maketitle
\begin{abstract}
Inequality measures such as the Gini coefficient are used to inform and motivate policymaking, and are increasingly applied to digital platforms.
We analyze how measures fare in pseudonymous settings that are common in the digital age.
A key challenge of such environments is the ability of actors to create fake identities under fictitious false names, also known as ``Sybils.''
While actors may do so to preserve privacy, we show that this can hamper inequality measurement: it is impossible for measures satisfying the literature's canonical set of desired properties to assess the inequality of an economy that harbors Sybils.
We characterize the class of all Sybil-proof measures, and prove they must satisfy relaxed versions of the established properties.
Furthermore, we show that the structure imposed restricts the ability to assess inequality at a fine-grained level.
We then apply our results to prove that popular measures are not Sybil-proof, with the famous Gini coefficient being but one example out of many.
Finally, we examine dynamics leading to the creation of Sybils in digital and traditional settings.
\end{abstract}

\section{Introduction}
The measurement of economic equality is a cornerstone of social science, providing tools and data to inform public discourse and motivate government policy \cite{connell2021inequality, auerbach2025public}.
Thus, it is unsurprising that inequality measures are used to study digital platforms \cite{nadler2022decentralized, feichtinger2024hidden}.
A famous example of such a measure is the \gls{GC} \cite{gini1912variabilita}, which ranges between a value of $0$ (if all individuals have equal wealth) and $1$ (if one entity holds all wealth).
Alarmingly, some found that the inequality of a prominent digital platform, Bitcoin, is ``like North Korea'' \cite{weisenthal2014how}, and ``more unequal than any other human society'' \cite{roubini2018bitcoin}, owing to a \gls{GC} of $0.88$ (see \cref{figure:InequalityComparison}).
Similar findings have been used to argue for a critical need to change the incentive mechanisms and governance structures of decentralized payment systems and \gls{DeFi} platforms to ensure they do not succumb to centralization \cite{gupta2018gini, nadler2022decentralized, sai2021characterizing, feichtinger2024hidden}.

\tikzsetnextfilename{InequalityComparison}
\begin{figure}
    \centering
    \scalebox{0.95}{\begin{tikzpicture}
    \begin{axis}[
        xlabel={Economy},
        ylabel={Gini Coefficient (GC)},
        xmajorgrids,
        ymajorgrids,
        xtick=data,
        table/col sep=comma,
        xticklabels from table={figures/InequalityComparison.csv}{name},
        xticklabel style={font=\tiny, rotate=90},
        xmin=-1,
        xmax=28.2,
    ]
        \addplot [scatter] table [x expr=\coordindex, col sep=comma, y=gini] {figures/InequalityComparison.csv};
    \end{axis}
    \end{tikzpicture}}
    \caption{Some argue that Bitcoin's \gls{GC} equals $0.88$ and is thus indicative of extreme inequality \cite{roubini2018bitcoin}, as this value positions it above any other global economy. Country-level data obtained from \cite{worldbank2025gini}.}
    \label{figure:InequalityComparison}
\end{figure}

However, applying classical measures directly to pseudonymous platforms overlooks a fundamental difference: the potential disconnect between observable identities and true economic actors.
This disconnect is not a mere technicality but a first-order feature of digital economies that invalidates the core assumptions of traditional inequality analysis.
Before we can answer how unequal these platforms are, we must first tackle a more fundamental question:
\emph{can we even measure their inequality reliably?}

\subsection{This Work}
This paper provides the first axiomatic answer to that question.
To do so, we establish the theoretical foundations for inequality measurement in the age of pseudonymity.

We show that reliable measurements are inherently challenging to carry out in pseudonymous settings, as common to digital platforms.
This is due to actors possibly holding their wealth across several Sybil identities, i.e., using multiple fictitious false names.
In fact, empirical work on measuring inequality in practice shows that such considerations are generally applicable: consider how individuals may hold wealth across seemingly disparate shell corporations to avoid taxes \cite{alstadsaeter2019tax, lustig2020missing, martinangeli2024inequality}.
Such behavior is common among sophisticated actors \cite{auerbach2000capital, gruber2002elasticity, slemrod2007cheating, gabaix2016dynamics, auerbach2025public}, making the setting applicable to non-pseudonymous markets.

To analyze the impact of Sybils, we use the literature's canonical set of properties desired from inequality measures \cite{cowell2011measuring}.
Thus, measures should receive a scalar vector representing the distribution of wealth over an economy's population, and ideally satisfy several axioms:
1. the transfer principle (i.e., given a population ordered by wealth, order-preserving wealth transfers from rich individuals to poorer ones reduce inequality),
2. symmetry (that is, invariance to permutations),
3. insensitivity to population size, and
4. scale independence (i.e., invariance under uniform scaling).

Before we elaborate on our results, we wish to illustrate that the impact of Sybils can be counterintuitive.
One might assume that splitting wealth across multiple accounts would generally decrease measured inequality.
That is, that observed inequality is a lower bound on the true one.
This intuition is \emph{false}, as we show in \cref{exmp:SybilsIncreaseInequality}.
\begin{example}[The Transfer Principle Implies Sybils Can Increase Inequality]
    \label{exmp:SybilsIncreaseInequality}
    Consider a wealth distribution $\left(10, 10\right)$, and suppose the second actor splits its wealth over two Sybil accounts, resulting in $\left(10, 9, 1\right)$.
    This Sybil manipulation actually increases inequality when considering measures satisfying the transfer principle, such as the \gls{GC}.
\end{example}

Having established this, we proceed to discuss our results.
We begin our analysis by highlighting the tension between Sybil-proofness and the literature's standard desired properties (\cref{sec:Impossibilities}).
Our first main result is an impossibility: no Sybil-proof inequality measure can satisfy the transfer principle.
While the transfer principle is a hallmark of inequality measures, one may be tempted to circumvent the impossibility by replacing it with a mild alternative.
Thus, we consider the egalitarian zero principle, which requires that all egalitarian wealth distributions (i.e., where all entries are identical) should correspond to the lowest possible inequality value, while any other distribution would have strictly larger inequality.
As we show, this again results in an impossibility.

Having delineated what is not possible, we map the design space for everything that is possible by fully characterizing the class of Sybil-proof measures (\cref{sec:Characterizing}).
As we show, this class is exactly the family of metrics which are \emph{sum-dependent}, meaning that given an input wealth distribution, its inequality is calculated solely based on the sum of wealth.
This is a natural companion to the impressive line of work characterizing decomposable measures \cite{bourguignon1979decomposable, shorrocks1980class, shorrocks1984inequality, vega2013axiomatic, heikkuri2025subgroup}, i.e., measures that, given aggregate information on two sub-populations, can calculate their union's inequality.
In a way, our sum dependent family considers a further condensed form of aggregation.
As we prove, this is not only a feature, but a necessity.

While we begin with our impossibilities and only then provide a complete characterization, it is possible, of course, to reverse this order: start with the characterization and derive the impossibilities as corollaries.
We have chosen this structure as this work's central finding is the problem itself: the tension between robustness to pseudonymity and the very properties that give inequality measures their meaning.
This allows us to derive the impossibilities using the literature's basic building blocks, without requiring the novel definitions and machinery used for the characterization.
This then positions the characterization as \emph{explaining} the structural limits leading to the negative results, rather than \emph{driving} them.

We culminate by showing our results' applicability to other settings and notable metrics.
In total, we offer a toolbox covering a large class of measures, including the \gls{GC}, \gls{CV}, \gls{AM}, \gls{GE} metrics, and the \gls{HH}, \gls{TL}, and \gls{TT} indices.

\paragraphNoSkip{Structure}
We begin by going over related work in \cref{sec:RelatedWork}, and proceed to formalize the setting in \cref{sec:Reasoning}.
Then, we present our impossibilities in \cref{sec:Impossibilities} and our characterization in \cref{sec:Characterizing}.
We culminate with a discussion on the generality of our results in \cref{sec:Discussion}.
Full proofs of all results and a comprehensive overview of the related literature are provided in \cref{sec:Proofs,sec:AdditionalRelatedWork}, correspondingly.

\section{Related Work}
\label{sec:RelatedWork}
To the best of our knowledge, our work is the first to provide a theoretic framework for inequality measurement under pseudonymity.
We contribute to several bodies of work, including the literature on the axiomatic approach to inequality measures, the limitations of such measures, and the analyses of both Sybils and collusion in economic settings.

\paragraphNoSkip{Inequality Measures}
The axiomatic study of inequality includes the foundational \cite{gini1912variabilita, dalton1920measurement, atkinson1970measurement}, which establish the core properties desired in measures.
Following are the elegant characterizations of \cite{bourguignon1979decomposable, foster1983axiomatic, shorrocks1987transfer}, that consider refinements and additions to the canonical desiderata, e.g., decomposability (a property which we discuss in \cref{sec:Discussion}).
Decomposable measures that are given the within-group and between-group inequalities of two populations, can ``staple'' the two together to compute their union's inequality.
Some recent work on the axiomatic approach includes the characterization of Theil inequality ordering by \cite{vega2013axiomatic}, the study of ``unit-consistent'' (i.e., scale independent, as in \cref{def:Scale}) indices by \cite{rio2010new}, and the analyses of decomposition carried out by \cite{nevescosta2019not, heikkuri2025subgroup}.
Work on the limitations of inequality measures includes analyses of sampling flaws, e.g., typing errors that affect the magnitude of sampled wealth \cite{cowell1996robustness}, incomplete representation of sub-populations \cite{korinek2006survey}, wealth under-reporting  \cite{lustig2020missing}, and inconsistent data collection methodology \cite{atkinson2001promise}.
We augment the existing substantial body of work by also considering Sybil-proofness and collusion-resistance in the context of inequality measures.
The comprehensive \cite{cowell2011measuring, stiglitz2025origins} present a summa of the topic, and the thorough \cite{yitzhaki2013gini} focuses on the \gls{GC}.

\paragraphNoSkip{Inequality and Pseudonymity}
Our theoretic framework nicely complements empirical measurement studies.
Thus, some have tried to empirically estimate inequality in blockchain protocols \cite{chemaya2025quantifying}, also while heuristically identifying Sybils \cite{fritsch2024analyzing}.
However, notable works note that methods are hard to validate due to missing ground-truth data \cite{victor2020address}, and that actors can create identities using cryptographic protocols that provide formal unlinkability guarantees \cite{wahrstatter2024blockchain}.
Other related studies consider social networks \cite{nielsen2013do}, open-source communities \cite{chelkowski2016inequalities}, peer-to-peer file-sharing \cite{hales2009bittorrent}, distributed payment systems \cite{ovezik2025sok}, and \gls{DeFi} platforms \cite{nadler2022decentralized}.

\paragraphNoSkip{Sybils and Collusion}
The theoretic literature on collusion and Sybils (sometimes called ``false names'') is extensive, and highlights the importance of these topics for digital settings \cite{conitzer2010using}.
However, no previous work considered inequality metrics.
The literature studied auctions \cite{yokoo2004effect}, voting \cite{wagman2008optimal}, file sharing protocols \cite{zohar2009adding}, querying mechanisms \cite{zhang2021sybil}, multi-level marketing \cite{drucker2012simpler}, consensus protocols \cite{chen2019axiomatic}, and transaction fee mechanisms \cite{roughgarden2024transaction}, among others.
Also related are foundational works on rights arbitration \cite{oneill1982problem} and bankruptcy \cite{aumann1985game}, where manipulations include introducing new creditors \cite{dagan1993bankruptcy}, with later work showing that some classic solutions are Sybil-proof or collusion-resistant \cite{belotti2021bankruptcy}.
While the context and axiomatic frameworks differ, the literature demonstrates that these manipulations are recurring challenges across settings.

\section{Reasoning About Hidden Economies}
\label{sec:Reasoning}
We consider the problem of measuring the inequality of an economy characterized by a disconnect between its underlying hidden wealth distribution among real actors and its observable distribution, as actors can report wealth using Sybil identities under false names.

\subsection{The Economy}
\paragraphNoSkip{Agents}
The economy is populated by $\wealthn$ real actors with a wealth distribution $\wealth \define \left( \wealth_1, \dots, \wealth_\wealthn \right)$, both $\wealthn, \wealth$ are hidden, and $\owealthn$ identities whose wealth distribution $\owealth \define \left( \owealth_1, \dots, \owealth_\owealthn \right)$ is observable.
Bridging the surface-level $\owealth$ and hidden $\wealth$ is an unobservable matrix $\report \in \mathbb{R}_{\ge 0}^{\wealthn \times \owealthn}$, where $\report_{i,j}$ is the fraction of wealth reported by actor $i$ under identity $j$.

\paragraphNoSkip{Manipulations}
We say a real actor $i \in \wealths$ is \emph{truthful} \gls{iff} it reports its wealth using exactly one identity and does not create any identity with no wealth.
However, actors need not be truthful.
Thus, $i \in \wealths$ can create \emph{Sybils} by distributing wealth across several identities.
Actors can also \emph{collude} and pool wealth together, defined as a set $\coalition \subseteq \wealths$ with $\lvert \coalition \rvert > 1$ and an identity $j \in \owealths$ where $\forall i \in \coalition: \report_{i,j} > 0$.
Finally, total observable wealth must match hidden wealth, per \cref{def:Conservation}.
This permits decreasing hidden wealth, e.g., due to ``burning'' \cite{gafni2025transaction}, if observed wealth decreases equally.
\begin{definition}[Wealth Conservation]
    \label{def:Conservation}
    We say that a report matrix $\report$ conserves an economy's wealth across its hidden and observable states \gls{iff} both the following hold:
    \begin{align*}
        \begin{cases}
        \forall i \in \wealths: 1 = \sum_{j \in \owealths} \report_{i,j}
        \\
        \forall j \in \owealths: \owealth_j = \sum_{i \in \wealths} \left( \report_{i,j} \wealth_i \right)
        \end{cases} 
    \end{align*}
\end{definition}

\subsection{Inequality Measure Axioms}
An inequality measure is a continuous function $\inequality: \bigcup_{k \in \mathbb{N}} \mathbb{R}_{\ge 0}^k \rightarrow \mathbb{R}$ mapping a wealth distribution, represented as a scalar vector of any size, to a scalar value.
The axiomatic inequality measure literature specifies a core set of desired properties \cite{rothschild1973some, fields1978inequality, cowell2011measuring, stiglitz2025some}.

First, measures should be insensitive to the economy's total amount of wealth or the units in which wealth is denominated, i.e., inequality should be invariant under uniform scaling.
This is captured by \cref{def:Scale}, which is attributed to \cite{atkinson1970measurement} by \cite{foster1983axiomatic}.
\begin{definition}[Scale Independence]
    \label{def:Scale}
    Measure $\inequality$ is \emph{scale independent} \gls{iff} for any $\alpha \in \mathbb{R}_{> 0}$ and any wealth distribution $x = \left(x_1, \dots, x_k \right) \in \mathbb{R}_{\ge 0}^k$ where $k \in \mathbb{N}$, then: $\inequality \left( x \right) = \inequality \left( \alpha \cdot x \right)$.
    Scalar multiplication is defined in the standard way, i.e., $\alpha \cdot x = \left( \alpha \cdot x_1, \dots, \alpha \cdot x_k \right)$.
\end{definition}
Inequality should be independent of population size, per \cref{def:Population}, which \cite{foster1983axiomatic} attributes to \cite{dalton1920measurement}.
Intuitively, a wealth distribution is compared to a counterfactual synthesized by duplicating it.
So, $\left(5, 10\right)$ should have inequality identical to $\left(5, 10, 5, 10\right)$.
\begin{definition}[Population Insensitivity]
    \label{def:Population}
    Measure $\inequality$ is \emph{population insensitive} \gls{iff} for any $k \in \mathbb{N}$ and any wealth distribution $x = \left(x_1, \dots, x_k \right) \in \mathbb{R}_{\ge 0}^k$, the inequality of $x$ is identical to that of $x \bigsqcup x$, i.e., $x$ concatenated with itself:
    \begin{align*}
        \inequality\left(x\right)
        =
        \inequality\left(x_1, \dots, x_k, x_1, \dots, x_k\right)
        =
        \inequality\left(x \bigsqcup x\right)
        .    
    \end{align*}    
\end{definition}
The axiom of symmetry (stated in \cref{def:Symmetry}), is assigned to \cite{dalton1920measurement} by \cite{atkinson1970measurement}.
This property postulates that inequality should be preserved under wealth distribution permutations.
E.g., both $\left( 5, 10 \right)$ and $\left( 10, 5 \right)$ should have the same inequality.
\begin{definition}[Symmetry]
    \label{def:Symmetry}
    Measure $\inequality$ is \emph{symmetric} \gls{iff} for any $k \in \mathbb{N}$, any permutation matrix $P \in \left\{0,1\right\}^{k \times k}$, and any wealth distribution $x = \left(x_1, \dots, x_k \right) \in \mathbb{R}_{\ge 0}^k$, we have that inequality is invariant to permutation:
    $\inequality \left( x \right) = \inequality \left( P x \right)$.
\end{definition}
The property given in \cref{def:TransferPrinciple}, sometimes called the Pigou-Dalton principle \cite{foster1983axiomatic} as it is traced to \cite{pigou1912wealth, dalton1920measurement}, is emblematic of inequality measures: equality should increase by a rank-preserving ``progressive'' transfer.
I.e., a strictly positive transfer where the sender is wealthier than the receiver both before and after the transfer.
\begin{definition}[Transfer Principle]
    \label{def:TransferPrinciple}
    Measure $\inequality$ satisfies the \emph{transfer principle} \gls{iff} for any $k \in \mathbb{N}_{\ge 2}$, any $x = \left(x_1, \dots, x_k \right) \in \mathbb{R}_{\ge 0}^k$, any $i, j \in \left[ k \right]$ with $x_i < x_j$, and any $\delta \in \left( 0, \frac{x_j - x_i}{2} \right]$, the inequality of $x$ is strictly greater than that of $x' = \left( x_1, \dots, x_i + \delta, \dots, x_j - \delta, \dots, x_k \right)$, which is identical to $x$ except a $\delta$ transfer from $j$ to $i$:
    $\inequality \left( x' \right) < \inequality \left( x \right).$
    Given $x, x'$ as before, the \emph{weak transfer principle} holds \gls{iff} post-transfer inequality is weakly lower: $\inequality \left( x' \right) \le \inequality \left( x \right).$
\end{definition}

\section{Impossibilities for Sybil-Proof Measures}
\label{sec:Impossibilities}
We now derive the first batch of our main results, showing that it is impossible for measures to withstand Sybil manipulations and simultaneously satisfy subsets of the canonical set of desired properties.
To do so, we enrich the canonical desiderata for inequality measures with \cref{def:SybilProof}, a new axiom that captures the essence of Sybil-resistance.
\begin{definition}[Sybil-Proofness]
    \label{def:SybilProof}
    A measure $\inequality$ is \emph{Sybil-proof} \gls{iff} it is invariant under any wealth conserving Sybil manipulation.
    That is, for any $\wealthn \in \mathbb{N}$, any hidden wealth distribution $\wealth = \left(\wealth_1, \dots, \wealth_\wealthn \right) \in \mathbb{R}_{\ge 0}^\wealthn$, and for any observable wealth distribution $\owealth = \left(\owealth_1, \dots, \owealth_\owealthn \right) \in \mathbb{R}_{\ge 0}^\owealthn$ generated from $\wealth$ via a wealth-conserving report matrix $\report \in \mathbb{R}_{\ge 0}^{\wealthn \times \owealthn}$ that involves Sybil manipulations, then we must have: $\inequality(\wealth) = \inequality(\owealth).$
\end{definition}
\cref{def:SybilProof} considers ``pure'' Sybil manipulations, strengthening our results by showing that they hold irrespective of collusion.
Moreover, we later show (e.g., in \cref{sec:Characterizing}) that the structural impact of Sybils on inequality measures \emph{necessitates} collusion-resistance.

\subsection{The Transfer Principle}
We now consider how the transfer principle (\cref{def:TransferPrinciple}) restricts robustness to Sybil manipulations.
We start with a warm-up, \cref{res:SybilImpossibilityScale}, to demonstrate the intricacies of the various axiomatic properties given in \cref{sec:Reasoning}.
Then, we follow by showing that an impossibility is obtainable using a minimal set of properties in \cref{res:SybilImpossibilityTransfer}.
\begin{restatable}{proposition}{resSybilImpossibilityScale}
    \label{res:SybilImpossibilityScale}
    No inequality measure can simultaneously satisfy the transfer principle, scale independence, population insensitivity, and Sybil-proofness.
\end{restatable}
Our proof delicately applies properties to produce wealth distributions that could feasibly be created with and without Sybil manipulations.
Thus, we first apply population insensitivity to double the population.
This doubles the economy's wealth, and so we use the scale independence property to halve the wealth back to the amount we started with.
Then, we perform a series of transfers where all but one of the ``original'' and ``doubled'' identities are reconciled back into a single identity.
This leaves us with two identities which are not reconciled.
To wrap-up the proof, we consider how transfers between these identities impact inequality, and, via the transfer principle, reach a contradiction to Sybil-proofness.

We continue with our first main result, \cref{res:SybilImpossibilityTransfer}.
\begin{restatable}{theorem}{resSybilImpossibilityTransfer}
    \label{res:SybilImpossibilityTransfer}
    No Sybil-proof measure satisfies the transfer principle.
\end{restatable}
To prove this, we start from some initial hidden wealth distribution, and show the observable distribution obtained from the truthful report is obtainable from a different hidden distribution and several distinct Sybil manipulations.
To get an impossibility, we ``tie'' the manipulations together by showing they can be obtained by performing a sequence of progressive transfers between distinct real actors, starting from the initial truthful distribution.
By the transfer principle, this must reduce inequality, reaching a contradiction.
While our proof considers general initial distributions, it can be illustrated via a simplified example.
\begin{example}[Concrete Demonstration of \cref{res:SybilImpossibilityTransfer}]
    Consider a single actor with all wealth $\left( 1 \right)$, and two possible observable distributions: $\left( 0.5, 0.5 \right)$ and $\left( 0.9, 0.1 \right)$.
    By the transfer principle, the latter distribution has strictly greater inequality.
    Sybil-proofness is violated: the two do not have identical inequality.
\end{example}

\subsection{Relaxing the Desiderata}
\label{sec:Relaxing}
One might hope that by relaxing this principle to its weak form, a meaningful Sybil-proof metric could be constructed.
However, the weak transfer principle permits metrics that cannot distinguish between distributions (as we show in \cref{sec:Characterizing}), thus highlighting one merit of the strict transfer principle.
To avoid that, we present in \cref{def:Egalitarian} a mild axiom called the egalitarian zero principle, which has been featured in prior work \cite{cowell2011measuring}.
\begin{definition}[Egalitarian Zero Principle]
    \label{def:Egalitarian}
    $\inequality$ satisfies the \emph{egalitarian zero principle} \gls{iff} for any $k \in \mathbb{N}$ and any distribution $x = \left( x_1, \dots, x_k \right)$, then $\inequality \left( x \right) = 0$ \gls{iff} $x$ is egalitarian ($\forall i_1, i_2 \in \left[k\right]: x_{i_1} = x_{i_2}$).
    The \emph{weak egalitarian zero principle} holds \gls{iff} for any egalitarian $x$: $\inequality \left( x \right) = 0$.
\end{definition}
Even in its non-weak form, the egalitarian zero axiom is a relaxed alternative to the transfer principle: a minimum is imposed on the measure, and, rather than requiring that \emph{any} progressive order-preserving transfer would lower inequality, this is only strictly enforced for transfers that result in full equality.
However, we show in \cref{res:SybilImpossibilityEgalitarian} that this relatively lax property still results in an impossibility.
\begin{restatable}{proposition}{resSybilImpossibilityEgalitarian}
    \label{res:SybilImpossibilityEgalitarian}
    No Sybil-proof inequality measure can satisfy the egalitarian zero principle.
\end{restatable}
We prove this by showing a non-egalitarian distribution which, for a specific Sybil manipulation, can be transformed into an egalitarian distribution.
This creates a conflict between the egalitarian zero principle, which requires the two corresponding inequalities to differ, and Sybil-proofness, which necessitates all inequalities to be the same.

\section{Characterizing Sybil-Proof Measures}
\label{sec:Characterizing}
So far, we have shown that Sybil-proofness cannot co-exist with some of the commonly sought-after properties of the literature.
We now turn to fully characterize which measures are resistant to Sybil manipulations, and the properties which they \emph{can} satisfy.
We do this by first presenting a series of results showing that Sybil-proofness in some cases also implies other properties.
Then, we prove that requiring Sybil-proofness and certain other mild properties necessarily imposes a specific structure on measures.
We demonstrate this by presenting a concrete Sybil-proof measure, and then generalize it into a family containing all Sybil-proof metrics which are either scale independent or population insensitive.
Finally, we characterize the entire class of Sybil-proof inequality measures, derive all properties that they satisfy, and prove that they admit constructions differing from the aforementioned family.

\subsection{Axioms Implied by Sybil-Proofness}
We now present a set of useful results, showing the implications of imposing Sybil-proofness.
E.g., in \cref{res:SybilSymmetric} we show that Sybil-proofness has a nice by-product: symmetry.
\begin{restatable}{proposition}{resSybilSymmetric}
    \label{res:SybilSymmetric}
    Any Sybil-proof measure is symmetric.
\end{restatable}
\cref{res:SybilSymmetric} serves as an illuminating companion to our first main result, \cref{res:SybilImpossibilityTransfer}, where we show that coupling Sybil-proofness with symmetry is directly at odds with the transfer principle.
In a way, our proof provides a new answer to Shakespeare's question, ``\emph{what's in a name?}''
That is because the symmetry property, also sometimes referred to as anonymity \cite{dasgupta1973notes}, is directly implied by Sybil-proofness: the presence of Sybil identities makes it necessary for a measure to be agnostic to permutations.
As we show, this extends further, and essentially requires measures to operate at a ``coarse'' level.

Next, we combine \cref{res:SybilSymmetric,res:SybilImpossibilityTransfer} to derive \cref{res:SybilWeakTransfer}, thus proving that Sybil-resistant measures must also adhere to the weak transfer principle.
Our proof shows that this is a delicate result, as one has to rule-out the possibility that progressive order-preserving transfers (per \cref{def:TransferPrinciple}) may \emph{harm} inequality.
To do so, we construct several Sybil manipulations where the wealth of real actors is chosen carefully as to ensure that one actor controls both identities that take part in transfers.
\begin{restatable}{proposition}{resSybilWeakTransfer}
    \label{res:SybilWeakTransfer}
    Any Sybil-proof inequality measure satisfies the weak transfer principle.
\end{restatable}

We proceed with \cref{res:SybilScalePopulation}, where we highlight the interplay between Sybil-proofness, scale independence and population insensitivity.
In particular, we show that invariance under population duplication, the condition at the heart of \cref{def:Population}, can be viewed as a Sybil manipulation which is combined with some scaling of the economy's wealth.
\begin{restatable}{proposition}{resSybilScalePopulation}
    \label{res:SybilScalePopulation}
    Any scale independent Sybil-proof metric is population insensitive.
\end{restatable}
This result is obtained by performing a Sybil manipulation where each identity splits its wealth equally across two identities, and then scaling the resulting distribution by $2$.

\subsection{Aggregation Invariance}
We prove an analogous result to \cref{res:SybilScalePopulation}, i.e., coupling Sybil-proofness with the population insensitivity property implies scale independence (see \cref{res:SybilPopulationScale}).
However, this requires stronger technical heavy-lifting and additional groundwork, that we now perform.
Fortunately, this very same groundwork results in several tools which are of immense use through the remainder of the paper.

We continue by presenting \cref{def:Aggregation}, and showing in \cref{res:SybilAggregation} that this axiom is equivalent to Sybil-proofness.
Intuitively, this axiom implies that a measure does not need to observe an economy's entire wealth distribution;
instead, it can observe a ``summarized'' distribution which is shorter by $1$, if the wealth of the omitted index of the original distribution is spread over the summarized distribution's indices.
\begin{definition}[Aggregation Invariance]
    \label{def:Aggregation}
    Measure $\inequality$ is \emph{aggregation invariant} \gls{iff} for any $x = \left(x_1, \dots, x_k \right) \in \mathbb{R}_{\ge 0}^k$ with $k \in \mathbb{N}_{\ge 2}$, and any aggregation matrix $\agg \in \mathbb{R}_{\ge 0}^{\left(k-1\right) \times k}$ such that $\sum_{i \in \left[k-1\right]} \left( \agg x \right)_i = \sum_{j \in \left[k\right]} x_j$, then: 
    $\inequality \left( x \right) = \inequality \left( \agg x \right).$
\end{definition}
\begin{restatable}{theorem}{resSybilAggregation}
    \label{res:SybilAggregation}
    $\inequality$ is aggregation invariant \gls{iff} it is Sybil-proof.
\end{restatable}
We prove the first direction (Sybil-proofness $\Rightarrow$ aggregation-invariance) by assuming some wealth distribution and aggregation matrix, and constructing a counterfactual distribution where a single actor initially controls all wealth.
This is followed by presenting two Sybil manipulations: one ending in the assumed distribution, and another that arrives at the result of applying the aggregation matrix to the assumed distribution.
We show the second direction by letting some hidden and observable wealth distributions, and inductively constructing a sequence of aggregation matrices that take both distributions to a counterfactual ``initial'' distribution where again only a single actor is present.

We follow with one of our main results, \cref{res:SybilPopulationConstant}, where we show that imposing population insensitivity on Sybil-proof measures must results in some constant measure.
\begin{restatable}{theorem}{resSybilPopulationConstant}
    \label{res:SybilPopulationConstant}
    If $\inequality$ is a population insensitive Sybil-proof inequality metric, then it must be expressible as a constant function.
    That is, there is some $c \in \mathbb{R}$ such that for any distribution $x = \left(x_1, \dots, x_k \right) \in \mathbb{R}_{\ge 0}^k$ where $k \in \mathbb{N}$: $\inequality \left( x \right) = c$.
\end{restatable}
The proof is by contradiction: given two distinct distributions resulting in different inequalities, both can be aggregated into each one's corresponding sum via \cref{res:SybilAggregation}.
Then, by repeated applications of Sybil-proofness and population insensitivity, we show that the measure does not only admit inequality-preserving population doubling, but also wealth halving.
Due to the literature's standard continuity assumption \cite{shorrocks1984inequality}, the inequality at the limit of any given starting wealth distribution must equal that of the singleton distribution with no wealth, i.e., $\inequality \left( 0 \right).$
This shows that, in fact, population insensitivity is quite a strong property when used in tandem with Sybil-proofness.
In our full characterization of Sybil-proofness, we also consider measures which are not population insensitive, and show that indeed the broad class of Sybil-proof measures admits other types of functions.

Having assembled the necessary tools, we prove the analog of \cref{res:SybilScalePopulation} where the roles of scale independence and population insensitivity are reversed by applying \cref{res:SybilPopulationConstant} and noting that constant metrics are scale independent.
\begin{restatable}{corollary}{resSybilPopulationScale}
    \label{res:SybilPopulationScale}
    Any population insensitive Sybil-proof measure is also scale independent.
\end{restatable}

\subsection{The Family of Constant Measures}
The previous results seem to point to the possibility that large swaths of the design space of Sybil-proof measures are occupied by metrics that are expressible using constant functions.
Before carrying out our full characterization of Sybil-proof metrics, we analyze one specific such measure (\cref{def:MeasureConstant}), and follow by considering the family of constant measures (\cref{def:FamilyConstant}).
We prove in \cref{res:DesiderataConstant} that one particular constant measure can satisfy the complete desiderata, where the transfer and egalitarian zero principles are in their weak form.
While the large set of properties may paint a seemingly rosy picture, the measure cannot be used to meaningfully compare two wealth distributions, being a constant function.
Our characterization of the constant measure family plays a key technical role in later results, as, by satisfying certain mild properties, any measure collapses to this form.
\begin{definition}[Baseline Measure $\baseline$]
    \label{def:MeasureConstant}
    Given any wealth distribution $x = \left(x_1, \dots, x_k \right) \in \mathbb{R}_{\ge 0}^k$ where $k \in \mathbb{N}$, the \emph{baseline measure} $\baseline$ is defined as:
    $\baseline \left( x \right) \define 0.$
\end{definition}
\begin{restatable}{proposition}{resDesiderataConstant}
    \label{res:DesiderataConstant}
    The measure $\baseline$ is scale independent, population insensitive, symmetric, Sybil-proof, and satisfies the weak egalitarian zero and weak transfer principles.
\end{restatable}
This result can be proven by going back to the definitions of the different axioms and seeing that, due to the measure being constant, these properties hold by definition.

Having observed that $\baseline$ satisfies the desiderata, including the weak egalitarian zero and transfer principles, consider that the specific choice that the metric would always output $0$ is, while helpful in satisfying the weak egalitarian zero principle, also somewhat arbitrary.
Indeed, $\baseline$ can be extended into a family of constant inequality metrics, as we do in \cref{def:FamilyConstant}.
We prove in \cref{res:FamilyConstant} that this family is notable for comprising all scale-independent and population insensitive Sybil-proof measures.
\begin{definition}[Constant Measure Family]
    \label{def:FamilyConstant}
    The \emph{constant measure family} is the set $C \define \left\{ \constant \, \lvert \, c \in \mathbb{R} \right\}$, where for all $k \in \mathbb{N}$ and $x = \left( x_1, \dots, x_k \right)$ then: $\forall \constant \in C: \constant \left( x \right) = c.$
\end{definition}
\begin{restatable}{corollary}{resFamilyConstant}
    \label{res:FamilyConstant}
    Any Sybil-proof measure $\inequality$ which is either scale independent or population insensitive must be part of the constant measure family and satisfy the symmetry and weak transfer principles.
    Moreover, $\inequality$ satisfies the weak egalitarian zero principle \gls{iff} $\inequality = \baseline$.
\end{restatable}
The proofs that each of the properties mentioned by \cref{res:FamilyConstant} are satisfied are corollaries of \cref{res:SybilSymmetric,res:SybilWeakTransfer,res:SybilScalePopulation,res:SybilPopulationConstant}, with the exception of the egalitarian zero principle, which is a corollary of \cref{res:DesiderataConstant}.

\subsection{Completing the Characterization}
To break away from the degenerate structure imposed by \cref{res:FamilyConstant}, we turn back to \cref{res:SybilAggregation}.
By following the technique applied in that result's proof to its logical conclusion, one can view any Sybil-proof measure as essentially a function which, given a wealth distribution vector, collates it into a single scalar value: its total wealth.
We formalize this intuition in \cref{def:MeasureSum,res:FamilySum}.
\begin{definition}[Sum Dependence]
    \label{def:MeasureSum}
    Measure $\inequality$ is \emph{sum-dependent} \gls{iff} for any distributions $x, y$ where $x = \left(x_1, \dots, x_{k_1} \right) \in \mathbb{R}_{\ge 0}^{k_1}$, $y = \left(y_1, \dots, y_{k_2} \right) \in \mathbb{R}_{\ge 0}^{k_2}$ and $k_1, k_2 \in \mathbb{N}$:
    $
    \inequality \left( x \right)
    =
    \inequality \left( y \right)
    \Leftrightarrow
    \sum_{i \in \left[k_1\right]} x_i
    =
    \sum_{j \in \left[k_2\right]} y_j
    .
    $
\end{definition}
Notably, measures satisfying \cref{def:MeasureSum} treat any distribution as if it originates from a collusion by the grand coalition of all actors.
This is a form of ``perfect'' coordination, akin to a world where all identities belong to one actor.
\begin{restatable}{corollary}{resFamilySum}
    \label{res:FamilySum}
    A measure is Sybil-proof \gls{iff} it is sum-dependent.
    Such measures are symmetric and satisfy the weak transfer principle, and those which are scale independent or population insensitive are constant measures.
\end{restatable}
This is proven by applying  \cref{res:SybilAggregation,res:FamilyConstant}.

An interesting implication of \cref{res:FamilySum} is that, because Sybil-proof measures are sum-dependent, then they essentially come with collusion-resistance baked-in.
That is because such measures must treat any observable distribution as generated by a single actor who may arbitrarily spread out wealth over many identities.
Thus, the inequality measure cannot distinguish Sybil manipulations from a perfectly colluding grand coalition comprising all of the economy's actors.
That is not to say that results are driven by that form of strong collusion.
In fact, some of the technical challenges faced thus far stem from the delicate task of optimizing the analyses to minimize the amount of actors required to carry out manipulations.
In fact, manipulations by a larger amount of actors are generally performed to ``set the stage'' and reach a state with well-defined inequality.
Upon reaching such a state, it is taken as some counterfactual initial state, which is then used to drive results.

The technique driving \cref{res:FamilySum} leads to another, quite strong, result.
\begin{restatable}{corollary}{resEgalitarinZeroConstant}
    \label{res:EgalitarinZeroConstant}
    A Sybil-proof measure $\inequality$ satisfies the weak egalitarian zero principle \gls{iff} $\inequality = \baseline$.
\end{restatable}
To prove the first direction ($\inequality$ is Sybil-proof and satisfies the weak egalitarian zero principle $\Rightarrow \inequality = \baseline$), consider some distribution and aggregate it recursively until reaching a ``singleton'' distribution with a single entry.
The distribution is egalitarian, so its inequality should equal $0$ per \cref{def:Egalitarian}.
By \cref{res:SybilAggregation}, a Sybil-proof measure is invariant under aggregation, and so the initial inequality also equals $0$.
The other direction ($\inequality$ is Sybil-proof and satisfies the weak egalitarian zero principle $\Leftarrow \inequality = \baseline$) is shown by \cref{res:DesiderataConstant}.

\section{Discussion}
\label{sec:Discussion}

In this work, we take the perspective of an on-looker, e.g., a tax auditor, attempting to assess the inequality in a given economy.
The on-looker has access to the wealth reports created by the population inhabiting the economy.
These reports, while faithfully presenting the magnitude of the population's underlying wealth, may be distorted due to actors spreading wealth across multiple Sybil identities.
That is, the difficulty lies not in observing the amount of wealth, but in assigning it to the correct entities, as the linking between observable identities and ``real'' actors is hidden.

We have fully characterized all Sybil-proof measures.
Our analysis reveals a fundamental tension between the traditional axiomatic foundations for inequality measures and their robustness to pseudonymity, which manifests in both digital and traditional settings.
For example, this tension is not alleviated by relaxing strict properties like the transfer principle (perhaps embodying the essence of inequality measures) or replacing them with mild alternatives such as the egalitarian zero principle.
The robustness of our results emphasize the challenges of devising Sybil-proof measures.
Possible interesting directions for future theoretical work include considering other relaxations, e.g., robustness to a bounded number of Sybils or providing approximate inequality guarantees in the presence of Sybils.
We note that empirical studies highlight the general difficulty of assessing the parameters required for such measures to succeed in practice \cite{victor2020address, messias2023airdrops}.

We now highlight the generality of our of results.

\subsection{Common Measures}
Our results are applicable to some of the most commonly-used measures.

\subsubsection*{Putting the Gini Back in the Bottle}
When discussing inequality measures, the \gls{GC} deserves special attention: it is widely used \cite{mehran1976linear, giorgi2017gini, patty2019measuring}, and has been called the ``default'' choice \cite{nevescosta2019not} and the ``most popular measure'' for studying inequality \cite{ferreira2020defense}.
Notably, the \gls{GC} was found to be more robust than others on data that may contain errors \cite{cowell1996robustness}.
However, \cref{res:GCSybilProof} shows it is not Sybil-proof.
To see why, we first state one of the canonical formulations of the \gls{GC} in \cref{def:GC}.
See \cite{yitzhaki2013gini} for other formulations.
\begin{definition}[\glsxtrfull{GC}]
    \label{def:GC}
    Given wealth distribution $x = \left( x_1, \dots, x_k \right)$ with $k \in \mathbb{N}$, the \emph{\gls{GC}} is defined as:
    $$
    \gc \left(x\right)
    \define
    \frac{
        \sum_{i \in \left[k\right]}
        \sum_{j \in \left[k\right]}
        \lvert x_i - x_j \rvert
    }{
        2k \sum_{i \in \left[k\right]} x_i
    }
    .
    $$
\end{definition}
\begin{restatable}{corollary}{resGCSybilProof}
    \label{res:GCSybilProof}
    The \gls{GC} is not Sybil-proof.
\end{restatable}
As can be gleaned from \cref{def:GC}, when applied to a strictly positive egalitarian distribution, the \gls{GC} equals zero and thus immediately satisfies the weak egalitarian zero principle.
This positions the \gls{GC} squarely within the realm of \cref{res:EgalitarinZeroConstant}.
However, as the \gls{GC} is not a constant measure (besides directly testing this using \cref{def:GC}, one can also observe \cref{figure:InequalityComparison}), it cannot be Sybil-proof.

\subsubsection*{Generalized Entropy Measures}
The \gls{GE} family of measures received attention by the foundational literature due to its convenient properties.
Chief among them is decomposability, which we state in \cref{def:Decomposability}, following the definition of \cite{shorrocks1984inequality}.
Intuitively, this property confers a measure the ability to estimate the inequality of a big group by considering the inequality of its subgroups and the subgroups' means and population sizes.
\begin{definition}[Decomposability]
    \label{def:Decomposability}
    Metric $\inequality$ is \emph{decomposable} \gls{iff} there is a function $D$ where for any $x = \left(x_1, \dots, x_{k_1} \right) \in \mathbb{R}_{\ge 0}^{k_1}$, $y = \left(y_1, \dots, y_{k_2} \right) \in \mathbb{R}_{\ge 0}^{k_2}$, and letting $\mu \left( x \right) = \frac{\sum_{i \in \left[k_1\right]} x_i}{k_1}$ and $\mu \left( y \right) = \frac{\sum_{j \in \left[k_2\right]} y_j}{k_2}$, then:
    $$
    \inequality \left( x \bigsqcup y \right)
    =
    D \left(
    \inequality \left( x \right),
    \mu \left( x \right),
    k_1,
    \inequality \left( y \right),
    \mu \left( y \right),
    k_2
    \right)
    .
    $$
\end{definition}
Results such as those of \cite{shorrocks1984inequality} show that the \gls{GC} is not decomposable, and that any scale independent decomposable measure belongs to the \gls{GE} family.
Adapting the statement of \cite{shorrocks1984inequality}, this is shown by \cref{res:GECharcterization}.
\begin{restatable}[\gls{GE} Family \cite{shorrocks1984inequality}]{theorem}{resGECharcterization}
    \label{res:GECharcterization}
    $\inequality$ is scale independent and decomposable \gls{iff} $\forall x = \left( x_1, \dots, x_k \right)$ there is a $c \in \mathbb{R}$ and a continuous function $F\left(\inequality, k\right)$ strictly increasing in its first parameter with $F\left(0,k\right) = 0$ such that:
    $$
    F \left( \inequality \left( x \right), k \right)
    \define
    \begin{cases}
    \frac{1}{k \cdot c(c-1)}
    \sum_{i \in \left[k\right]} \left( \left( \frac{x_i}{\mu \left( x \right)} \right)^c - 1 \right),
    \text{ if } c \ne 0, 1
    \\
    \frac{1}{k} \sum_{i \in \left[k\right]} \frac{x_i}{\mu \left( x \right)} \ln \left( \frac{x_i}{\mu \left( x \right)} \right),
    \text{ if } c = 1
    \\
    \frac{1}{k} \sum_{i \in \left[k\right]} \ln \left( \frac{\mu \left( x \right)}{x_i} \right),
    \text{ if } c = 0
    \end{cases}
    $$
    $\inequality$ is also population invariant \gls{iff} $F$ is independent of $k$.
\end{restatable}
The literature shows that this characterization and its transformations admit a long list of measures, including the \gls{CV}, the \gls{AM}, and the \gls{HH}, \gls{TL}, and \gls{TT} indices \cite{cowell2011measuring}.
However, we prove that this family is not Sybil-proof.
\begin{restatable}{corollary}{resGESybilProof}
    \label{res:GESybilProof}
    The \gls{GE} family is not Sybil-proof.
\end{restatable}
One can verify that none of these measures is constant, yet, due to their scale independence, we obtain from \cref{res:FamilySum} that they cannot be Sybil-proof.

\subsection{Sybils in Traditional and Digital Economies}
Well-known strategies used to obscure wealth in conventional economies and digital ones serve as real-world analogs to the Sybil manipulations that we consider.
We now discuss the implications of our results, in light of previous work studying both types of economies.

\subsubsection*{Traditional Economies}
The difficulty of assigning wealth to its owner is well-documented in traditional economies: individuals may utilize complex legal structures, such as shell corporations and irrevocable trusts, to legally distance themselves from assets \cite{auerbach2000capital, gruber2002elasticity, slemrod2007cheating, gabaix2016dynamics, alstadsaeter2019tax, lustig2020missing, martinangeli2024inequality, auerbach2025public}.
Such methods serve the same function as Sybil identities, obscuring the origins of observable wealth from onlookers.

\subsubsection*{Digital Economies}
In the digital world, identities may be easily hidden, and, in some cases, obscurity can even be the default modus operandi.
For example, on \cite{bitcoin.org2025protect}, a website originally created by Bitcoin's creator \cite{bitcoin.org2025about}, an explicit directive is given:
\begin{quote}
    \emph{``To protect your privacy, you should use a new Bitcoin address each time you receive a new payment.
    Additionally, you can use multiple wallets for different purposes.
    Doing so allows you to isolate each of your transactions in such a way that it is not possible to associate them all together.
    People who send you money cannot see what other Bitcoin addresses you own and what you do with them.
    This is probably the most important advice you should keep in mind.''}
\end{quote}
That need not even be said for actors who require the utmost obscurity.
A prominent example is given by so-called airdrop farmers.
These are specialized actors who may create many Sybil identities to maximize their profits from rewards issued by blockchain-based platforms \cite{yaish2025platform}.
Platforms tend to give out per-capita rewards in an effort to create or expand their user-bases, and so, airdrop farmers attempt to ensure their different identities cannot be linked to each other.
Some rewards are of \emph{governance} tokens that enable voting on a platform's direction, with several works noting that such designs put platforms at risk \cite{feichtinger2024sok, dotan2023vulnerable}.

\printbibliography[heading=bibintoc]

\appendix

\section{Proofs}
\label[appendix]{sec:Proofs}

\resSybilImpossibilityScale*
\begin{proof}  
    Assume an inequality measure $\inequality$ that has all of the properties given in the statement, and considering some initial observable wealth distribution $\owealth = \left( \owealth_1, \dots, \owealth_\owealthn \right)$ with an initial inequality $\inequality^{init}$.
    By population insensitivity, concatenating the distribution with itself produces a distribution with same inequality under the assumed measure.
    Next, one can halve all wealth with the inequality remaining equal to $\inequality^{init}$, owing to the measure being scale independent.
    This is followed by choosing an identity with a strictly positive amount of wealth, and, if there are multiple options, pick the one with the minimal index $j$.
    
    By construction, $j \le \owealthn$, and identity $j+\owealthn$ has the same wealth as $j$, and thus also commands a strictly positive amount of wealth: $\frac{1}{2} \owealth_j = \frac{1}{2} \owealth_{j+\owealthn} > 0$.
    Now, iterate over each $j' \ne j$ where $j' \le \owealthn$, and, at every step, transfer all the wealth of identity $j'+\owealthn$ to $j'$.
    This culminates in a wealth distribution which is identical to the initial one, except: 1. identity $j$ has half the wealth, 2. there are additional $\owealthn-1$ identities with no wealth whatsoever, and 3. there is an extra identity at index $j+\owealthn$ with wealth equal to $\frac{1}{2} \owealth_j$.
    
    Consider a transfer of the entirety of identity $j + \owealthn$'s wealth to identity $j$, after which $j$ has a wealth equal to:
    \begin{align*}
        \frac{1}{2} \owealth_j + \frac{1}{2} \owealth_{j+\owealthn}
        =
        \frac{1}{2} \owealth_j + \frac{1}{2} \owealth_j
        =
        \owealth_j
        >
        0
    \end{align*}
    On the other hand, $j+\owealthn$ has none, and denote the resulting wealth distribution by $s_1$, and the corresponding inequality by $\inequality^{s_1}$.
    Following on that transfer, consider another transfer in the opposite direction, where $\frac{1}{4} \owealth_j$ is moved from $j$ to $j+\owealthn$, which maintains the wealth ranking between the two identities, as $j$ remains with $\frac{3}{4} \owealth_j$ wealth while $j+\owealthn$ has $\frac{1}{4} \owealth_j$, and because $\owealth_j > 0$ then we have that $\frac{3}{4} \owealth_j > \frac{1}{4} \owealth_j$.
    Denote the resulting wealth distribution and corresponding inequality by $s_2$ and $\inequality^{s_2}$.
    As this is a transfer that satisfies the standard version of \cref{def:TransferPrinciple}, it must be that $\inequality^{s_2} < \inequality^{s_1}$.
    However, one can construct a Sybil manipulation that starts from the initial state and results in the wealth distribution $s_1$, and, similarly, a different manipulation can start from the initial state and culminate in $s_2$.
    As $\inequality$ is assumed to be Sybil-proof, the inequalities for both $s_1$ and $s_2$ must be equal.
    Thus, we obtain a contradiction.

    While we have finished proving the result, we find it worthwhile to dwell more about the technicalities involved.
    Observe that the elaborate process that we went through in the beginning of the proof was performed to lay the technical groundwork for the last two manipulations.
    Those two, in particular, are driving the resulting impossibility, and, both can be performed using a single agent.
\end{proof}

\resSybilImpossibilityTransfer*
\begin{proof}
    Consider a measure $\inequality$ satisfying the properties stated in \cref{res:SybilImpossibilityTransfer},and a hidden wealth distribution $\wealth = \left( \wealth_1, \wealth_2 \right)$ where $\wealth_1 = 1-\epsilon$ for some small $\epsilon < \frac{1}{4}$, and $\wealth_2 = \epsilon$.
    
    Consider a truthful report (i.e., the identity matrix), and two transfers. 
    First, transfer an amount of wealth equal to $\frac{1}{2} - 2 \epsilon$ from identity $1$ to identity $2$.
    This transfer satisfies \cref{def:TransferPrinciple}, and so should strictly lower inequality.
    Denote the initial inequality by $\inequality^{t_1}$, and the outcome's inequality by $\inequality^{t_2}$.
    Follow on this transfer with another one from $1$ to $2$, this time, of $\epsilon$, resulting in equalizing the wealth of $1$ and $2$.
    As before, the resulting inequality, denoted by $\inequality^{t_3}$, should be strictly lower:
    \begin{align*}
        \inequality^{t_1}
        >
        \inequality^{t_2}
        >
        \inequality^{t_3}
        .
    \end{align*}

    Next, consider the counterfactual initial hidden distribution where a single actor has all wealth: $\wealth = (1)$.
    Furthermore, consider three Sybil manipulations where the single actor creates two identities: in the first, the resulting observable distribution is $\owealth^{s_1} = \left( 1-\epsilon, \epsilon \right)$, in the second it is $\owealth^{s_2} = \left( \frac{1}{2}+\epsilon, \frac{1}{2}-\epsilon \right)$, and in the last: $\owealth^{s_3} = \left( \frac{1}{2}, \frac{1}{2} \right)$.
    Denote the inequalities corresponding to the manipulations correspondingly by $\inequality^{s_1}, \inequality^{s_2}, \inequality^{s_3}$.
    From the Sybil-proofness property, we have that all are equal:
    \begin{align*}
        \inequality^{s_1}
        =
        \inequality^{s_2}
        =
        \inequality^{s_3}
        .
    \end{align*}

    For the measure to be self-consistent, it must be that $\forall i \in \left\{ 1, 2, 3 \right\}: \inequality^{t_i} = \inequality^{s_i}$, reaching a contradiction:
    \begin{align*}
        \inequality^{s_1}
        =
        \inequality^{t_1}
        \ne
        \inequality^{t_2}
        =
        \inequality^{s_2}
        =
        \inequality^{t_2}
        \ne
        \inequality^{t_3}
        =
        \inequality^{s_3}
        .
    \end{align*}
    Essentially, the transfer principle forces $\inequality$ to consider implicit ``Sybil-to-Sybil'' transfers as bona fide ``progressive'' transfers, and so the Sybil-proofness property cannot hold.
\end{proof}

\resSybilImpossibilityEgalitarian*
\begin{proof}
    Assume towards contradiction that there is an inequality measure $\inequality$ as stated.
    Let:
    \begin{align*}
        \begin{cases}
            \wealth_1 = 5 = \owealth_1
            \\
            \wealth_2 = 10 = \owealth_2
        \end{cases}
    \end{align*}
    Consider the wealth distribution $\owealth = \left( \owealth_1, \owealth_2 \right)$.
    As $\owealth$ is non-egalitarian and because the measure satisfies the egalitarian zero principle, it must be that $\inequality\left(\owealth\right) > 0$.
    Next, consider a Sybil manipulation where the wealth of identity $2$ is split across two identities, resulting in:
    \begin{align*}
        \owealth'
        =
        \left( \owealth'_1, \owealth'_2, \owealth'_3 \right)
        =
        \left( 5, 5, 5 \right)
        .
    \end{align*}
    As it is egalitarian, then by the assumption we started from: $\inequality\left( \owealth' \right) = 0$.
    Because the measure is Sybil-proof, we obtain a contradiction:
    \begin{align*}
        0
        <
        \inequality \left( \owealth \right)
        =
        \inequality \left( \wealth \right)
        =
        \inequality \left( \owealth' \right)
        =
        0
        .
    \end{align*}
\end{proof}

\resSybilSymmetric*
\begin{proof}
    Assume a Sybil-proof inequality measure $\inequality$, some wealth distribution $\owealth = \left( \owealth_1, \dots, \owealth_\owealthn \right)$, and some permutation matrix $P \in \left\{0,1\right\}^{k \times k}$.
    Consider a counterfactual ``initial'' wealth distribution where the entire population comprises just one actor.
    Moreover, that actor has wealth equal to $\wealth_1 = \sum_{j \in \owealths} \owealth_j$.
    Denote $\wealth = \left( \wealth_1 \right)$, and perform a Sybil manipulation where the single actor's entire wealth is split into $\owealthn$ identities such that we obtain the assumed distribution, i.e., $\owealth = \left( \owealth_1, \dots, \owealth_\owealthn \right)$.
    By the Sybil-proofness of $\inequality$, it must be that $\inequality \left( \wealth \right) = \inequality \left( \owealth \right)$.
    Next, again owing to the measure being Sybil-proof, the Sybil manipulation where the single actor spreads its wealth such that the resulting distribution equals $P \owealth$ is inequality preserving: $\inequality \left( \wealth \right) = \inequality \left( P \owealth \right)$.
    In total, we get that:
    \begin{align*}
        \inequality \left( \owealth \right)
        =
        \inequality \left( \wealth \right)
        =
        \inequality \left( P \owealth \right)
        .
    \end{align*}
    As $P$ is a general permutation matrix, the same process can be applied for any other one, implying that $\inequality$ is symmetric.
\end{proof}

\resSybilWeakTransfer*
\begin{proof}
    Let $\inequality$ be a Sybil-proof measure.
    From \cref{res:SybilSymmetric}, it must be symmetric.
    So, \cref{res:SybilImpossibilityTransfer} implies that $\inequality$ cannot satisfy the regular transfer principle.
    To show that the measure must satisfy the weak transfer principle, we thus need to prove that a transfer which satisfies the weak principle's conditions cannot increase inequality.
    
    Assume towards contradiction that there is a wealth distribution $\owealth = \left( \owealth_1, \dots, \owealth_\owealthn \right)$ that fits the criteria of \cref{def:TransferPrinciple}, yet that a progressive rank-preserving transfer culminating in distribution $\owealth' = \left( \owealth'_1, \dots, \owealth'_\owealthn \right)$ results in a larger inequality.
    Denote the sender of the transfer by $j_1$, and the recipient by $j_2$.
    Next, consider a counterfactual initial wealth distribution $\wealth$ where real actor $1$ has a wealth of $\wealth_1 = \owealth_{j_1} + \owealth_{j_2}$, and, append to $\wealth$ the wealth of the remaining $\owealthn-2$ identities which were not involved in the transfer.
    Then, apply a Sybil manipulation that transforms the resulting distribution, $\wealth$, into the post-transfer distribution, $\owealth'$.
    By the assumption that $\inequality$ is Sybil-proof, we have that this manipulation is inequality-preserving: $\inequality \left( \wealth \right) = \inequality \left( \owealth' \right)$.
    Finally, apply another Sybil manipulation, but, this time, $\wealth$ is transformed into $\owealth$, the pre-transfer distribution.
    As before, this cannot lead to changes in inequality, i.e., $\inequality \left( \wealth \right) = \inequality \left( \owealth \right)$.
    We obtain a contradiction:
    $$
    \inequality \left( \owealth \right) = \inequality \left( \wealth \right) = \inequality \left( \owealth' \right).$$
\end{proof}

\resSybilScalePopulation*
\begin{proof}
    Assume that $\inequality$ is scale independent and consider some wealth distribution $\owealth = \left( \owealth_1, \dots, \owealth_\owealthn \right)$.
    Perform a Sybil manipulation where the wealth of each identity $j$ is spread evenly across two identities: 1. $j$, and 2. a new identity $j + \owealthn$.
    The resulting distribution is:
    \begin{align*}
        \left( \frac{1}{2} \owealth_1, \dots, \frac{1}{2} \owealth_\owealthn, \frac{1}{2} \owealth_1, \dots, \frac{1}{2} \owealth_\owealthn \right).
    \end{align*}
    As $\inequality$ is Sybil-proof, this manipulation is inequality-preserving.
    As the measure is assumed to be scale independent, the inequality remains unchanged after doubling each identity's wealth.
    The final wealth distribution we reached after all manipulations is $\left( \owealth_1, \dots, \owealth_\owealthn, \owealth_1, \dots, \owealth_\owealthn \right)$, and, this is equal to $\owealth \bigsqcup \owealth$.
    As inequality did not change throughout the operations we performed, then $\inequality \left( \owealth \right) = \inequality \left( \owealth \bigsqcup \owealth \right)$, satisfying \cref{def:Population}.
\end{proof}

\resSybilAggregation*
\begin{proof}
    We examine each direction separately.

    \paragraphNoSkip{Direction 1: Sybil-Proofness $\Rightarrow$ Aggregation-Invariance}
    Assume $\inequality$ is Sybil-proof, and let $\owealth = \left(\owealth_1, \dots, \owealth_\owealthn \right) \in \mathbb{R}_{\ge 0}^\owealthn$ with $\owealthn \in \mathbb{N}_{\ge 2}$ be some observable wealth distribution.
    To prove this direction holds, we need to show that \cref{def:Aggregation} is satisfied, so, let $\agg \in \mathbb{R}_{\ge 0}^{\owealthn \times \left(\owealthn-1\right)}$ be some aggregation matrix.
    We proceed by constructing a counterfactual ``initial'' wealth distribution where a single actor is in possession of all wealth represented by $\owealth$.
    Thus, let $\wealth_1 = \sum_{j \in \left[ \owealthn \right]} \owealth_j$, and denote $\wealth = \left( \wealth_1 \right)$.
    In what follows, we use this initial hidden state as a baseline, and carry out two Sybil manipulations.
    In the first, the single hidden actor spreads its wealth over $\owealthn$ identities so that the resulting distribution is exactly $\owealth$.
    For the second, the actor distributes its wealth among $\owealthn-1$ identities to match $\agg \owealth$.
    By the measure's Sybil-proofness, inequality is preserved across both manipulations:
    \begin{align*}
        \inequality \left( \owealth \right)
        =
        \inequality \left( \wealth \right)
        =
        \inequality \left( \agg \owealth \right)
        .
    \end{align*}

    \paragraphNoSkip{Direction 2: Sybil-Proofness $\Leftarrow$ Aggregation-Invariance}
    Assume an aggregation-invariant inequality measure $\inequality$, and a tuple of wealth-preserving hidden and observable wealth distributions $\wealth = \left(\wealth_1, \dots, \wealth_\wealthn \right) \in \mathbb{R}_{\ge 0}^\wealthn$ with $\owealthn \in \mathbb{N}$ and $\owealth = \left(\owealth_1, \dots, \owealth_\owealthn \right) \in \mathbb{R}_{\ge 0}^\owealthn$ with $\owealthn \in \mathbb{N}$.
    To show $\inequality$ is Sybil-proof, we need to prove that $\inequality \left( \wealth \right) = \inequality \left( \owealth \right)$.
    To do so, we progressively aggregate $\wealth$ and $\owealth$, until each one is ``fully condensed'' and represents just a single identity.
    This can be performed, for example, by denoting $k = \max\left( \wealthn, \owealthn \right)$, and letting a sequence of aggregation matrices $\agg_2 \in \mathbb{R}^{2 \times 1}_{\ge 0}, \dots, \agg_{k} \in \mathbb{R}^{k \times \left( k - 1 \right)}_{\ge 0}$.
    Because by construction we have that both $k \ge \wealthn$ and $k \ge \owealthn$, we can now progressively multiply the matrices, obtaining the sum of each distribution: $\left( \agg_2 \dots \agg_\wealthn \wealth \right) \in \mathbb{R}$, and $\left( \agg_2 \dots \agg_\owealthn \owealth \right) \in \mathbb{R}$.
    As $\inequality$ is aggregation-invariant, one can obtain two sequences of equivalencies by induction:
    \begin{align*}
        \begin{cases}        
        \inequality \left( \wealth \right)
        =
        \inequality \left( \agg_\wealthn \wealth \right)
        =
        \dots
        =
        \inequality \left( \agg_2 \dots \agg_\wealthn \wealth \right)
        \\
        \inequality \left( \owealth \right)
        =
        \inequality \left( \agg_\owealthn \owealth \right)
        =
        \dots
        =
        \inequality \left( \agg_2 \dots \agg_\owealthn \owealth \right)
        \end{cases}
    \end{align*}
    Recall that \cref{def:Aggregation} ensures that aggregation matrices are sum-preserving, so because the sequence of multiplications is wholly performed using aggregation matrices, then:
    \begin{align*}
        \begin{cases}
        \sum_{i \in \wealths} \wealth_i
        =
        \left( \agg_2 \dots \agg_\wealthn \wealth \right)
        \\
        \sum_{j \in \owealths} \owealth_i
        =
        \left( \agg_2 \dots \agg_\owealthn \owealth \right)
        \end{cases}
    \end{align*}
    As, by definition, $\owealth$ is the outcome of a wealth-preserving Sybil manipulation performed from the baseline hidden distribution $\wealth$, then we have:
    \begin{align*}
        \sum_{i \in \wealths} \wealth_i
        &
        =
        \agg_2 \dots \agg_\wealthn \wealth
        \nonumber\\&
        =
        \agg_2 \dots \agg_\owealthn \owealth
        \nonumber\\&
        =
        \sum_{j \in \owealths} \owealth_j
        .
    \end{align*}
    Finally, we can combine this with our two inductive sequences of equalities, showing that indeed $\inequality$ is Sybil-proof:
    \begin{align}
        \label{eq:SybilAggregation}
        \inequality \left( \wealth \right)
        &
        =
        \inequality \left( \agg_\wealthn \wealth \right)
        =
        \dots
        =
        \inequality \left( \agg_2 \dots \agg_\wealthn \wealth \right)
        \nonumber\\&
        =
        \inequality \left( \sum_{i \in \wealths} \wealth_i \right)
        =
        \inequality \left( \sum_{j \in \owealths} \owealth_j \right)
        =
        \nonumber\\&
        =
        \inequality \left( \agg_2 \dots \agg_\owealthn \owealth \right)
        =
        \dots
        =
        \inequality \left( \agg_\owealthn \owealth \right)
        =
        \inequality \left( \owealth \right)
        .
    \end{align}
\end{proof}

\resSybilPopulationConstant*
\begin{proof}
    Let $\inequality$ be a Sybil-proof inequality measure which satisfies \cref{def:Population}, and assume towards contradiction that it not constant.
    Thus, there must be two distributions $\owealth^{s_1} \ne \owealth^{s_2}$ for which $\inequality$ reports differing inequalities: $\inequality(\owealth^{s_1}) \ne \inequality(\owealth^{s_2})$.
    If one distribution is longer than the other, then we can condense it until their lengths are matching without affecting the resulting inequality: this is possible because from \cref{res:SybilAggregation} we know that $\inequality$ is aggregation invariant, by virtue of being Sybil-proof.
    We can now fully condense both distributions to obtain each one's sum, and, again, by definition (i.e., \cref{def:Aggregation}), the resulting ``summaries'' have the same inequality as the original distributions: $\inequality \left( \owealth_1^{s_1} \right) = \inequality \left( \owealth^{s_1} \right) \ne \inequality \left( \owealth^{s_2} \right) = \inequality \left( \owealth_1^{s_2} \right)$, implying that: $\owealth_1^{s_1} \ne \owealth_1^{s_2}$.
    As wealth distributions are strictly positive, by extension their summations are also strictly positive.
    Now, consider the following two cases.
    
    \paragraphNoSkip{Case I: $\inequality$ is Scale Independent}
    Then: $\inequality \left( \frac{\owealth_1^{s_1}}{\owealth_1^{s_2}} \owealth^{s_2} \right) = \inequality \left( \owealth^{s_2} \right)$.
    So, we reach a contradiction:
    \begin{align*}
        \inequality \left( \owealth_1^{s_1} \right)
        =
        \inequality \left( \frac{\owealth_1^{s_1}}{\owealth_1^{s_2}} \owealth^{s_2} \right)
        =
        \inequality \left( \owealth^{s_2} \right)
        =
        \inequality \left( \owealth_1^{s_2} \right)
        .
    \end{align*}

    \paragraphNoSkip{Case II: Otherwise}
    If the measure is not scale independent, then there is some $\alpha \in \mathbb{R}_{> 0}$ together with $x$ such that $\inequality \left( x \right) \ne \inequality \left( \alpha x \right)$.
    It is not necessarily the case that this is true for all $\alpha, x$.
    If we are fortunate, there are $\alpha, x$, where $\alpha = \frac{\owealth_1^{s_1}}{\owealth_1^{s_2}}, x = \owealth^{s_2}$ or $\alpha = \frac{\owealth_1^{s_2}}{\owealth_1^{s_1}}, x = \owealth^{s_1}$, and the preceding case can be followed to obtain a contradiction.
    If that is not the case, then we follow a different route.
    
    Recall that per \cref{sec:Reasoning}, we have that $\inequality$ is continuous, with this being a standard assumption in the literature \cite{shorrocks1980class, foster1983axiomatic, shorrocks1984inequality}.
    In contrast, the current proof relies on a much weaker condition than that made by the literature: continuity is only required at $0$.
    Next, note that due to Sybil-proofness, then when considering the counterfactual where all wealth is controlled by a single actor, then the actor can spread its wealth equally across to Sybil identities without changing the inequality:
    $\inequality \left( \owealth^{s_1} \right) = \inequality \left( \frac{1}{2} \owealth^{s_1}, \frac{1}{2} \owealth^{s_1} \right)$.
    From population independence, it must be that:
    $\inequality \left( \frac{1}{2} \owealth^{s_1}, \frac{1}{2} \owealth^{s_1} \right) = \inequality \left( \frac{1}{2} \owealth^{s_1} \right)$.
    If that is not the case, then concatenating $\frac{1}{2} \owealth^{s_1}$ with itself does not preserve inequality, in contradiction to the assumption that $\inequality$ is population insensitive.
    Thus, we can again split $\frac{1}{2} \owealth^{s_1}$ evenly over two identities without changing inequality, due to Sybil-proofness:
    $\inequality \left( \frac{1}{2} \owealth^{s_1} \right) = \inequality \left( \frac{1}{4} \owealth^{s_1}, \frac{1}{4} \owealth^{s_1} \right).$
    As before, due to population insensitivity, we have:
    $\inequality \left( \frac{1}{4} \owealth^{s_1}, \frac{1}{4} \owealth^{s_1} \right) = \inequality \left( \frac{1}{4} \owealth^{s_1} \right).$
    We can repeat the process $k$ times, and, due to continuity at $0$:
    \begin{align*}
        \inequality \left( \owealth^{s_1} \right)
        =
        \lim_{k \rightarrow \infty} \inequality \left( \frac{\owealth^{s_1}}{2^k} \right)
        =
        \inequality \left( 0 \right)
        .
    \end{align*}
    This can be done with any wealth distribution, so for $c \define \inequality \left( 0 \right)$, then:
    $
    \forall x:
    \inequality \left( x \right)
    =
    c
    .
    $
\end{proof}

\resSybilPopulationScale*
\begin{proof}
    Let $\alpha > 0$.
    A Sybil-proof population insensitive $\inequality$ satisfies the conditions of \cref{res:SybilPopulationConstant}, and thus, one can express it using a constant function.
    As that is the case, for any wealth distribution $x = \left(x_1, \dots, x_k \right) \in \mathbb{R}_{\ge 0}^k$ where $k \in \mathbb{N}$ we have that: $\inequality \left( x \right) = \inequality \left( \alpha x \right)$.
    Given that we started from an arbitrary $\alpha$ per \cref{def:Scale}, then $\inequality$ is scale independent.
\end{proof}

\resDesiderataConstant*
\begin{proof}
    We prove that each axiom is satisfied.

    \paragraphNoSkip{Sybil-Proofness}
    The measure is Sybil-proof: for any two wealth distributions $x^{s_1}$ and $x^{s_2}$, we have that $\baseline \left( x^{s_1} \right) = \baseline \left( x^{s_2} \right)$, regardless of whether one is obtained via a Sybil manipulation from the other or not.

    \paragraphNoSkip{Weak Egalitarian Zero}
    By \cref{def:MeasureConstant}, $\baseline$ \emph{always} outputs $0$, whether or not its input is represents a completely equal wealth distribution.
    Thus, the current property is satisfied.
    
    \paragraphNoSkip{Scale Independence}
    As the measure is constant, scaling any wealth distribution cannot change the measure's output, implying that $\baseline$ is scale independent.
    
    \paragraphNoSkip{Symmetry, Population Insensitivity, and the Weak Transfer Principle}
    These properties are obtained by applying \cref{res:SybilSymmetric,res:SybilScalePopulation,res:SybilWeakTransfer}.
\end{proof}

\resFamilyConstant*
\begin{proof}
    As $\inequality$ is Sybil-proof, symmetry and the weak transfer principle are satisfied due to \cref{res:SybilSymmetric,res:SybilWeakTransfer}, correspondingly.
    Moreover, \cref{res:SybilScalePopulation,res:SybilPopulationConstant} precisely show that if $\inequality$ is either scale independent or population insensitive, it must be a member of the family of constant measures.
    Lastly, if $\inequality = \baseline$ then \cref{res:DesiderataConstant} implies that the measure satisfies the weak egalitarian zero principle, while the other direction is due to $\inequality \ne \baseline$ implying that there is some $c \ne 0$ such that for any wealth distribution $x = \left(x_1, \dots, x_k \right) \in \mathbb{R}_{\ge 0}^k$ where $k \in \mathbb{N}$: $\inequality \left( x \right) = c \ne 0$, thus, $\inequality$ cannot satisfy the egalitarian zero principle.
\end{proof}

\resFamilySum*
\begin{proof}
    The first part of the result's statement (i.e., a measure is Sybil-proof \gls{iff} it is sum-dependent) is an implication of \cref{res:SybilAggregation}, and, even more specifically, of \cref{eq:SybilAggregation}.
    The remaining parts are obtained by applying \cref{res:FamilyConstant}.
\end{proof}

\resEgalitarinZeroConstant*
\begin{proof}
    To prove the second direction, we need to show that the baseline measure $\baseline$ satisfies the weak egalitarian zero principle.
    Thankfully, we have done this in \cref{res:DesiderataConstant}.

    We continue by considering the remaining direction.
    Thus, assume that $\inequality$ is some Sybil-proof measure satisfying the weak egalitarian zero principle.
    Let $x = \left( x_1, \dots, x_k \right)$ be some wealth distribution.
    From \cref{res:SybilAggregation}, we can repeatedly aggregate $x$ without impacting inequality.
    So, if $x^a$ is the result of the $a$-th aggregation, then:
    \begin{align*}
        \inequality\left( x \right)
        =
        \inequality\left( x^1 \right)
        =
        \dots
        =
        \inequality\left( x^k \right)
        .
    \end{align*}
    In particular, we have that after the $k$-th aggregation, the corresponding wealth distribution comprises a single identity with all the economy's wealth: $x^k = \left( \sum_{j \in \left[k\right]} x_j \right)$.
    Because there are no other identities, $x^k$ is egalitarian, and we have: $\inequality \left( x^k \right) = 0$.
    As we have shown that inequality is maintained throughout the aggregation process, we have:
    \begin{align*}
        \inequality\left( x \right)
        =
        \inequality\left( x^1 \right)
        =
        \dots
        =
        \inequality\left( x^k \right)
        =
        0
        .
    \end{align*}
    Given that our starting point was some general $x$, then the inequality is necessarily $0$ for any distribution, implying that $\inequality = \baseline$.
\end{proof}

\resGCSybilProof*
\begin{proof}
    To show that the \gls{GC} is not Sybil-proof, we first prove that it satisfies the weak egalitarian zero principle (\cref{res:GCWeakEgalitarian}).
    \begin{restatable}{proposition}{resGCWeakEgalitarian}
        \label{res:GCWeakEgalitarian}
        The \gls{GC} satisfies the weak egalitarian zero principle.
    \end{restatable}
    This is proven by substituting an egalitarian wealth distribution into \cref{def:GC}.
    This result is obtained by direct arithmetic: scaling by any constant is cancelled out as the constant appears in both the numerator and the denominator of \cref{def:GC}.

    We proceed by showing in \cref{res:GCNotConstant} that the measure is not the constant zero measure.
    To see why this is true, one can perform the necessary calculations to obtain that $\gc \left( 1, 0 \right) \ne \gc \left( 1, 1 \right).$
    \begin{restatable}{proposition}{resGCNotConstant}
        \label{res:GCNotConstant}
        The \gls{GC} is not the constant zero measure $\baseline$.
    \end{restatable}

    Finally, we tie the above results using \cref{res:EgalitarinZeroConstant}, which states that any Sybil-proof measure satisfying the weak egalitarian zero principle must be the constant $0$ measure, i.e., $\baseline$.
    As the measures differ, the \gls{GC} cannot be Sybil-proof.
\end{proof}

\resGCWeakEgalitarian*
\begin{proof}
    Let $x = \left( x_1, \dots, x_k \right)$ be some egalitarian wealth distribution, e.g., $\forall i: x_i = c$ for some $c > 0$.
    Then, by substituting into \cref{def:GC}:
    \begin{align*}
        \gc \left(x\right)
        &
        =
        \frac{
            \sum_{i \in \left[k\right]}
            \sum_{j \in \left[k\right]}
            \lvert x_i - x_j \rvert
        }{
            2k \sum_{i \in \left[k\right]} x_i
        }
        \nonumber\\&
        =
        \frac{
            \sum_{i \in \left[k\right]}
            \sum_{j \in \left[k\right]}
            \lvert c - c \rvert
        }{
            2k \sum_{i \in \left[k\right]} c
        }
        \nonumber\\&
        =
        \frac{
            \sum_{i \in \left[k\right]}
            \sum_{j \in \left[k\right]}
            0
        }{
            2k \sum_{i \in \left[k\right]} c
        }
        \nonumber\\&
        =
        0
    \end{align*}
    Thereby satisfying the weak version of \cref{def:Egalitarian}.
\end{proof}

\resGCNotConstant*
\begin{proof}
    As we have shown in \cref{res:GCWeakEgalitarian} that the \gls{GC} satisfies the weak egalitarian zero principle, we have that $\gc \left( 1, 1 \right) = 0$.
    From \cref{def:GC}, we also have that for $x = \left( 1, 0 \right)$:
    \begin{align*}
        \gc \left( x \right)
        &
        =
        \frac{
            \sum_{i \in \left[k\right]}
            \sum_{j \in \left[k\right]}
            \lvert x_i - x_j \rvert
        }{
            2k \sum_{i \in \left[k\right]} x_i
        }
        \nonumber\\&
        = 
        \frac{
            \left(
                \lvert 1 - 1 \rvert
                +
                \lvert 1 - 0 \rvert
            \right)
            +
            \left(
                \lvert 0 - 1 \rvert
                +
                \lvert 0 - 0 \rvert
            \right)
        }{
            2 \cdot 2 \cdot \left( 1 + 0 \right)
        }
        \nonumber\\&
        = 
        \frac{
            1 + 1
        }{
            4
        }
        \nonumber\\&
        = 
        \frac{1}{2}
    \end{align*}
    And so, the \gls{GC} is not constant and also obtains some strictly positive values over part of its domain.
    Thus, the \gls{GC} is not $\baseline$:
    \begin{align*}
        \gc \left( 1, 0 \right)
        =
        \frac{1}{2}
        \ne
        0
        =
        \gc \left( 1, 1 \right)
    \end{align*}
\end{proof}

\resGESybilProof*
\begin{proof}
    Assume towards contradiction that the there is some member of the \gls{GE} family that is Sybil-proof.
    Recall that each member of this family is defined by a parameter $c$.
    As \cref{res:GECharcterization} shows that members of the \gls{GE} family are scale independent, then \cref{res:FamilySum} implies that the Sybil-proof member must also belong to the family of constant functions.
    So, to prove that this measure cannot be Sybil-proof, we need to show that it cannot be constant.
    We separately prove this for each possible value of $c$, which, per \cref{res:GECharcterization}, defines a concrete expression for the measure.
    \begin{restatable}{proposition}{resGENotConstantCZero}
        \label{res:GENotConstantCZero}
        The \gls{GE} family is not constant for $c  = 0$.
    \end{restatable}
    \begin{restatable}{proposition}{resGENotConstantCOne}
        \label{res:GENotConstantCOne}
        The \gls{GE} family is not constant for $c  = 1$.
    \end{restatable}
    \begin{restatable}{proposition}{resGENotConstantCRest}
        \label{res:GENotConstantCRest}
        The \gls{GE} family is not constant for $c  \ne 0, 1$.
    \end{restatable}
    All results are obtained by substituting concrete distributions into \cref{res:GECharcterization}.
\end{proof}

\resGENotConstantCZero*
\begin{proof}
    From \cref{res:GECharcterization}, given some $k$ and a wealth distribution $x = \left( x_1, \dots, x_k \right)$ of length $k$, then $\mu \left( x \right)$ is defined as the average of $x$, and we have that for $c = 0$:
    \begin{align*}
        F \left( \inequality \left( x \right), k \right)
        &
        =
        \frac{1}{k} \sum_{i \in \left[k\right]} \ln \left( \frac{\mu \left( x \right)}{x_i} \right)
        .
    \end{align*}
    Now, evaluate $F$ on $x = \left( 1, 3 \right)$:
    \begin{align*}
        F \left( \inequality \left( x \right), 2 \right)
        &
        =
        \frac{1}{k} \sum_{i \in \left[k\right]} \ln \left( \frac{\mu \left( x \right)}{x_i} \right)
        \nonumber\\&
        =
        \frac{1}{2} \sum_{i \in \left[2\right]} \ln \left( \frac{2}{x_i} \right)
        \nonumber\\&
        =
        \frac{1}{2} \left(
        \ln \left( \frac{2}{1} \right)
        +
        \ln \left( \frac{2}{3} \right)
        \right)
        \nonumber\\&
        \approx
        0.287
        .
    \end{align*}
    Next, evaluate the wealth distribution $y = \left( 1, 5 \right)$:
    \begin{align*}
        F \left( \inequality \left( y \right), 2 \right)
        &
        =
        \frac{1}{k} \sum_{i \in \left[k\right]} \ln \left( \frac{\mu \left( y \right)}{y_i} \right)
        \nonumber\\&
        =
        \frac{1}{2} \sum_{i \in \left[2\right]} \ln \left( \frac{3}{y_i} \right)
        \nonumber\\&
        =
        \frac{1}{2} \left(
        \ln \left( \frac{3}{1} \right)
        +
        \ln \left( \frac{3}{5} \right)
        \right)
        \nonumber\\&
        \approx
        0.587
        .
    \end{align*}
    Given the above calculations, we find that at least two different values are obtained:
    \begin{align*}
        F \left( \inequality \left( x \right), 2 \right)
        \approx
        0.287
        \ne 
        0.587
        \approx
        F \left( \inequality \left( y \right), 2 \right)
        .
    \end{align*}
\end{proof}

\resGENotConstantCOne*
\begin{proof}
    Examine the case where $c = 1$, as given by \cref{res:GECharcterization}:
    \begin{align*}
        F \left( \inequality \left( x \right), k \right)
        &
        =
        \frac{1}{k} \sum_{i \in \left[k\right]} \frac{x_i}{\mu \left( x \right)} \ln \left( \frac{x_i}{\mu \left( x \right)} \right)
        .
    \end{align*}
    We proceed by simplifying the above expression.
    Start by taking the first average term out of the summation:
    \begin{align*}
        F \left( \inequality \left( x \right), k \right)
        &
        =
        \frac{1}{k} \sum_{i \in \left[k\right]} \frac{x_i}{\mu \left( x \right)} \ln \left( \frac{x_i}{\mu \left( x \right)} \right)
        \nonumber\\&
        =
        \frac{1}{k \mu \left( x \right)} \sum_{i \in \left[k\right]} x_i \ln \left( \frac{x_i}{\mu \left( x \right)} \right)
        .
    \end{align*}
    Next, due to the basic logarithm property that $\log_a \left( \frac{b}{c} \right) = \log_a \left( b \right) - \log_a \left( c \right)$:
    \begin{align*}
        F \left( \inequality \left( x \right), k \right)
        &
        =
        \frac{1}{k \mu \left( x \right)} \sum_{i \in \left[k\right]} x_i \ln \left( \frac{x_i}{\mu \left( x \right)} \right)
        \nonumber\\&
        =
        \frac{1}{k \mu \left( x \right)}
        \sum_{i \in \left[k\right]} x_i \left(
        \ln \left( x_i \right)
        -
        \ln \left( \mu \left( x \right) \right)
        \right)
        .
    \end{align*}
    In total, we can write the following simplification:
    \begin{align}
        \label{eq:GEOneSimple}
        F \left( \inequality \left( x \right), k \right)
        &
        =
        \frac{1}{k \mu \left( x \right)}
        \sum_{i \in \left[k\right]} x_i \left(
        \ln \left( x_i \right)
        -
        \ln \left( \mu \left( x \right) \right)
        \right)
        .
    \end{align}
    For any $x = \left( 1, \dots, 1 \right)$ of length $k$, we have an average of $\mu \left( x \right) = 1$, so from \cref{eq:GEOneSimple}:
    \begin{align*}
        F \left( \inequality \left( x \right), k \right)
        &
        =
        \frac{1}{k \mu \left( x \right)}
        \sum_{i \in \left[k\right]} x_i \left(
        \ln \left( x_i \right)
        -
        \ln \left( \mu \left( x \right) \right)
        \right)
        \nonumber\\&
        =
        \frac{1}{k \cdot 1}
        \sum_{i \in \left[k\right]} \left( 1 \cdot \left(
        \ln \left( 1 \right)
        -
        \ln \left( 1 \right)
        \right)
        \right)
        \nonumber\\&
        =
        \frac{1}{k} \sum_{i \in \left[k\right]} \left( 0 - 0 \right)
        \nonumber\\&
        =
        \frac{1}{k} \left( 0 + \dots + 0 \right)
        \nonumber\\&
        =
        0
        .
    \end{align*}
    Let $e$ be Euler's number, i.e., $e \define \lim_{n \rightarrow \infty} \left( 1 + \frac{1}{n} \right)^n \approx 2.71828$.
    We follow by considering a wealth distribution equal to $y = \left( \frac{1}{2} e, \frac{3}{2} e \right)$.
    In particular:
    \begin{align*}
        \mu (y)
        &
        =
        \frac{1}{2} \left( \frac{1}{2} e + \frac{3}{2} e \right)
        =
        \frac{1}{2} \left( 2 e \right)
        =
        e
        .
    \end{align*}
    So, we obtain the following from \cref{eq:GEOneSimple}:
    \begin{align*}
        F \left( \inequality \left( y \right), 2 \right)
        &
        =
        \frac{1}{k \mu \left( y \right)}
        \sum_{i \in \left[k\right]} y_i \left(
        \ln \left( y_i \right)
        -
        \ln \left( \mu \left( y \right) \right)
        \right)
        \nonumber\\&
        =
        \frac{1}{2 \cdot e}
        \sum_{i \in \left[k\right]} y_i \left(
        \ln \left( y_i \right)
        -
        \ln \left( e \right)
        \right)
        \nonumber\\&
        =
        \frac{1}{2 e}
        \sum_{i \in \left[k\right]} y_i \left(
        \ln \left( y_i \right)
        -
        1
        \right)
        \nonumber\\&
        =
        \frac{1}{2 e}
        \left(
        \frac{1}{2} e \left( \ln \left( \frac{1}{2} e \right) - 1 \right)
        +
        \frac{3}{2} e \left( \ln \left( \frac{3}{2} e \right) - 1 \right)
        \right)
        \nonumber\\&
        =
        \frac{1}{4}
        \left(
        \ln \left( \frac{1}{2} e \right)
        -
        1
        +
        3 \left( \ln \left( \frac{3}{2} e \right) - 1 \right)
        \right)
        \nonumber\\&
        =
        \frac{1}{4}
        \left(
        \ln \left( \frac{1}{2} e \right)
        +
        3 \ln \left( \frac{3}{2} e \right)
        -
        4
        \right)
        \nonumber\\&
        =
        \frac{1}{4}
        \left(
        \ln \left( \frac{1}{2} e \right)
        +
        3 \ln \left( \frac{3}{2} e \right)
        \right)
        -1
        .
    \end{align*}
    By the basic properties of logarithms, we know that $\log_a \left( b \cdot c \right) = \log_a \left( b \right) + \log_a \left( c \right)$:
    \begin{align*}
        F \left( \inequality \left( y \right), 2 \right)
        &
        =
        \frac{1}{4}
        \left(
        \ln \left( \frac{1}{2} e \right)
        +
        3 \ln \left( \frac{3}{2} e \right)
        \right)
        -1
        \nonumber\\&
        =
        \frac{1}{4}
        \left(
        \ln \left( \frac{1}{2} \right)
        +
        \ln \left( e \right)
        +
        3 \left(
        \ln \left( \frac{3}{2} \right)
        +
        \ln \left( e \right)
        \right)
        \right)
        -1
        \nonumber\\&
        =
        \frac{1}{4}
        \left(
        \ln \left( \frac{1}{2} \right)
        +
        1
        +
        3 \left(
        \ln \left( \frac{3}{2} \right)
        +
        1
        \right)
        \right)
        -1
        \nonumber\\&
        =
        \frac{1}{4}
        \left(
        \ln \left( \frac{1}{2} \right)
        +
        3 \ln \left( \frac{3}{2} \right)
        +
        4
        \right)
        -1
        \nonumber\\&
        =
        \frac{1}{4}
        \left(
        \ln \left( \frac{1}{2} \right)
        +
        3 \ln \left( \frac{3}{2} \right)
        \right)
        +1-1
        \nonumber\\&
        =
        \frac{1}{4}
        \left(
        \ln \left( \frac{1}{2} \right)
        +
        3 \ln \left( \frac{3}{2} \right)
        \right)
        .
    \end{align*}
    Again, because of the basic logarithm property $\log_a \left( \frac{b}{c} \right) = \log_a \left( b \right) - \log_a \left( c \right)$:
    \begin{align*}
        F \left( \inequality \left( y \right), 2 \right)
        &
        =
        \frac{1}{4}
        \left(
        \ln \left( \frac{1}{2} \right)
        +
        3 \ln \left( \frac{3}{2} \right)
        \right)
        \nonumber\\&
        =
        \frac{1}{4}
        \left(
        \left( \ln(1) - \ln(2) \right)
        +
        3 \left( \ln(3) - \ln(2) \right)
        \right)
        \nonumber\\&
        =
        \frac{1}{4}
        \left(
        \left( 0 - \ln(2) \right)
        +
        3 \left( \ln(3) - \ln(2) \right)
        \right)
        \nonumber\\&
        =
        \frac{3}{4} \ln(3)
        - \ln(2)
        \nonumber\\&
        \approx
        0.13
        .
    \end{align*}
    And, we obtain:
    \begin{align*}
        F \left( \inequality \left( x \right), k \right)
        =
        0
        \ne 
        0.13
        \approx
        F \left( \inequality \left( y \right), 2 \right)
        .
    \end{align*}
    For a concrete instance of $x$, note that the above is true for all natural $k$, including $k = 2$, which results in $x = (1, 1).$
\end{proof}

\resGENotConstantCRest*
\begin{proof}
    Let $c \ne 0, 1$.
    Then, for such values of $c$, we have from \cref{res:GECharcterization} that:
    \begin{align*}
        F \left( \inequality \left( x \right), k \right)
        &
        =
        \frac{1}{k} \frac{1}{c(c-1)} \sum_{i \in \left[k\right]} \left( \left( \frac{x_i}{\mu \left( x \right)} \right)^c - 1 \right)
        .
    \end{align*}
    This can be simplified by first, moving the $-1$ term out of the summation:
    \begin{align*}
        F \left( \inequality \left( x \right), k \right)
        &
        =
        \frac{1}{k} \frac{1}{c(c-1)} \sum_{i \in \left[k\right]} \left( \left( \frac{x_i}{\mu \left( x \right)} \right)^c - 1 \right)
        \nonumber\\&
        =
        \frac{1}{k} \frac{1}{c(c-1)} 
        \left(
        \left( \sum_{i \in \left[k\right]} \left( \frac{x_i}{\mu \left( x \right)} \right)^c \right)
        - k
        \right)
        .
    \end{align*}
    Then, moving the $\frac{1}{k}$ inside the brackets:
    \begin{align*}
        F \left( \inequality \left( x \right), k \right)
        &
        =
        \frac{1}{k} \frac{1}{c(c-1)} 
        \left(
        \left( \sum_{i \in \left[k\right]} \left( \frac{x_i}{\mu \left( x \right)} \right)^c \right)
        - k
        \right)
        \nonumber\\&
        =
        \frac{1}{c(c-1)}
        \left(
        \left( \frac{1}{k} \sum_{i \in \left[k\right]} \left( \frac{x_i}{\mu \left( x \right)} \right)^c \right)
        - 1
        \right)
        .
    \end{align*}
    And, finally, noting that the summation can be simplified:
    \begin{align*}
        F \left( \inequality \left( x \right), k \right)
        &
        =
        \frac{1}{c(c-1)}
        \left(
        \left( \frac{1}{k} \sum_{i \in \left[k\right]} \left( \frac{x_i}{\mu \left( x \right)} \right)^c \right)
        - 1
        \right)
        \nonumber\\&
        =
        \frac{1}{c(c-1)} \left(
        \frac{
            \sum_{i \in \left[k\right]} \left( x_i \right)^c
        }{
            k \left( \mu \left( x \right) \right)^c
        }
        - 1
        \right)
        .
    \end{align*}
    So, we can write:
    \begin{align}
        \label{eq:GERestSimple}
        F \left( \inequality \left( x \right), k \right)
        &
        =
        \frac{1}{c(c-1)} \left(
        \frac{
            \sum_{i \in \left[k\right]} \left( x_i \right)^c
        }{
            k \left( \mu \left( x \right) \right)^c
        }
        - 1
        \right)
        .
    \end{align}
    Let $k$, and denote by $x$ the ``unit'' egalitarian wealth distribution of length $k$, i.e.:
    \begin{align*}
        x
        =
        \underbrace{\left( 1, \dots, 1 \right)}_{k \text{ times}}
        .
    \end{align*}
    Then, the distribution's average is $\mu\left(x\right) = 1$, and from \cref{eq:GERestSimple} we have that:
    \begin{align*}
        F \left( \inequality \left( x \right), k \right)
        &
        =
        \frac{1}{c(c-1)} \left(
        \frac{
            \sum_{i \in \left[k\right]} \left( x_i \right)^c
        }{
            k \left( \mu \left( x \right) \right)^c
        }
        - 1
        \right)
        \nonumber\\&
        =
        \frac{1}{c(c-1)} \left(
        \frac{
            \sum_{i \in \left[k\right]} \left( 1 \right)^c
        }{
            k \left( 1 \right)^c
        }
        - 1
        \right)
        \nonumber\\&
        =
        \frac{1}{c(c-1)} \left(
        \frac{
            k
        }{
            k
        }
        - 1
        \right)
        \nonumber\\&
        =
        \frac{1}{c(c-1)} \left(
        1
        - 1
        \right)
        \nonumber\\&
        =
        0
        .
    \end{align*}
    Next, consider $y = \left( 1, 3 \right)$, and substitute in \cref{eq:GERestSimple}:
    \begin{align*}
        F \left( \inequality \left( y \right), 2 \right)
        &
        =
        \frac{1}{c(c-1)} \left(
        \frac{
            \sum_{i \in \left[k\right]} \left( y_i \right)^c
        }{
            k \left( \mu \left( y \right) \right)^c
        }
        - 1
        \right)
        \nonumber\\&
        =
        \frac{1}{c(c-1)} \left(
        \frac{
            \sum_{i \in \left[2\right]} \left( y_i \right)^c
        }{
            2 \left( 2 \right)^c
        }
        - 1
        \right)
        \nonumber\\&
        =
        \frac{1}{c(c-1)} \left(
        \frac{
            \left( y_1 \right)^c + \left( y_2 \right)^c
        }{
            2^{c+1}
        }
        - 1
        \right)
        \nonumber\\&
        =
        \frac{1}{c(c-1)} \left(
        \frac{
            1^c + 3^c
        }{
            2^{c+1}
        }
        - 1
        \right)
        \nonumber\\&
        =
        \frac{1}{c(c-1)} \left(
        \frac{
            1 + 3^c
        }{
            2^{c+1}
        }
        - 1
        \right)
        .
    \end{align*}
    Note that we have $F \left( \inequality \left( y \right), 2 \right) = 0$ only when $c = 0, 1$, because:
    \begin{align*}
        F \left( \inequality \left( y \right), 2 \right)
        \Leftrightarrow
        \frac{
            1 + 3^c
        }{
            2^{c+1}
        }
        = 1
        \Leftrightarrow
        1 + 3^c = 2^{c+1}
        .
    \end{align*}
    Thus, we obtain that:
    \begin{align*}
        \forall c \ne 0, 1:
        F \left( \inequality \left( x \right), k \right)
        =
        0
        \ne 
        F \left( \inequality \left( y \right), 2 \right)
        .
    \end{align*}
    This completes the proof.
    For a specific example of a wealth distribution $x$, one can set $k=2$ and obtain $x = (1, 1).$
\end{proof}

\section{Additional Related Work}
\label[appendix]{sec:AdditionalRelatedWork}

\subsection*{Inequality in Traditional Settings}
The economic study of inequality measures is grounded in decades of research, including theoretical and empirical works such as \cite{gini1912variabilita, atkinson1970measurement, cowell1996robustness, blackorby1999income, giorgi1999income, atkinson2001promise, cowell2011measuring, yitzhaki2013gini, patty2019measuring, nevescosta2019not, connell2021inequality, auten2024income}.
While these works do not consider Sybil identities, some focus on somewhat related issues, such as measurement errors.

Among these, \cite{atkinson2001promise} highlights that comparing the inequality of different economies may be sensitive to methodological data collection gaps, e.g., some countries may report inequality calculated using pre-tax income, while others may use disposable income.
Furthermore, \cite{cowell1996robustness} analyzes the robustness of inequality estimation given samples drawn from some probability distribution, where, with an infinitesimal probability, a sample may have an arbitrarily large error (e.g., due to a typing error).
Then, the empirical influence function \cite{hampel1974influence} is used to show that measures satisfying the transfer principle (\cref{def:TransferPrinciple}) and decomposability (\cref{def:Decomposability}) may be biased by errors.
While we considerably differ, it is worthwhile to dwell on the complementarity of our works:
1. We consider a fluid population size that may change when actors create multiple identities.
2. Our impossibilities hold even when there is no error in the magnitude of reported wealth.
3. We also derive results for measures that only satisfy the weak transfer principle.
The focus of \cite{korinek2006survey} is the empirical estimation of the probability that households with different income would respond to income surveys, and how this may affect measured inequality.
The survey of \cite{lustig2020missing} presents various measurement pitfalls such as incomplete coverage and under-reporting of wealth and income.

\subsection*{Inequality in Digital Settings}
The literature on inequality in digital platforms, where pseudonymity is common and thus the threat of Sybil identities and false names is significant \cite{conitzer2010using}, includes works on social networks \cite{nielsen2013do, vachuska2025digital, machado2025super}, open-source communities \cite{chelkowski2016inequalities, petryk2023data, hoffmann2024value}, file-sharing communities \cite{hales2009bittorrent, rahman2010improving}, online healthcare platforms \cite{hsu2022understanding}, distributed consensus protocols \cite{ovezik2025sok, alzayat2021modeling, kondor2014do, lin2021measuring, gupta2018gini, gochhayat2020measuring, sai2021characterizing}, payment channels \cite{zabka2022short}, and \gls{DeFi} platforms \cite{nadler2022decentralized, fritsch2024analyzing, simons2021grants, barbereau2022defi, savelyev2022how, barbereau2022decentralised, feichtinger2024hidden}.
Of these, \cite{fritsch2024analyzing, ovezik2025sok, he2025six} attempt to empirically identify Sybil identities using heuristic methods, which some note as hard to validate owing to a lack of ground-truth \cite{victor2020address}.
It is worthwhile to stress that actors may mask the connections (if any) between the identities they create by using sophisticated cryptographic protocols.
Notably, the formal unlinkability guarantees of such protocols led to government-issued sanctions against them, their operators, and even their users \cite{wahrstatter2024blockchain, yaish2024speculative}.
Some have highlighted the governance structure of blockchain platforms as particularly susceptible to concentration, which can result in risk to user funds and to platform stability and security \cite{dotan2023vulnerable, messias2023understanding, yaish2024strategic, austgen2023dao, feichtinger2024sok, kitzler2025governance, yang2025decentralization, celig2025distributional}.

\subsection*{Sybils and Collusion}
The impact of false names and Sybils has been studied in many economic settings, e.g., auctions \cite{yokoo2001robust, sakurai2003false, yokoo2004effect, yokoo2010false, iwasaki2010worst, alkalayhoulihan2014false, gafni2020vcg, gafni2023optimal, pan2024sibyl}, voting \cite{wagman2008optimal, waggoner2012evaluating, gafni2021worst}, facility location \cite{sonoda2016false}, matching \cite{todo2013false}, social networks \cite{friedman2001social, jia2023incentivising}, querying mechanisms \cite{chen2018sybil, zhang2021sybil, zhang2023collusion}, data sharing for machine learning \cite{gafni2022long}, peer-to-peer file-sharing \cite{zohar2009adding, cheng2022study, cheng2024tight}, decentralized exchanges \cite{zhang2024rediswap}, and promotional reward schemes, including both multi-level marketing for ``traditional'' settings \cite{emek2011mechanisms, drucker2012simpler, lv2013fair}, and blockchain airdrops \cite{messias2023airdrops, yaish2024tierdrop, yaish2025platform}.
This is joined by the literature on collusion \cite{samuelson1958exact, stigler1964theory, marshall2007bidder, micali2007collusion, yokoo2004effect, marshall2007bidder, wagman2008optimal, bachrach2010honor, drucker2012simpler, feige2013cascade, zhang2021sybil, zhang2023collusion, branzei2023learning}.

A line of work highlights Sybil-proofness as a core property of mechanisms operating in pseudonymous settings.
This includes, for example, the literature on permissionless consensus protocols \cite{eyal2014majority, babaioff2012bitcoin, chen2019axiomatic, budish2024economic}, i.e., protocols which facilitate agreement between a network of distrusting nodes.
Moreover, Sybil-resistance is a part of the canonical desiderata for transaction fee mechanisms used by decentralized payment systems, which determine which transactions system nodes should process, the fees transactions pay, and the revenue nodes can collect \cite{lavi2019redesigning, yao2018incentive, gafni2022greedy, chung2023foundations, shi2023what, roughgarden2024transaction, ferreira2024incentive, gafni2024barriers, chung2024collusion, gafni2024discrete, bahrani2024transaction, yaish2024incentives, ganesh2024revisiting, chen2025bayesian, wang2025proophi, garimidi2025transaction, gafni2025transaction}.

\section{Reproduction}
\label[appendix]{sec:Reproduction}
We ensure the full reproducibility of our paper's contents by providing all the code and data used to produce them, together with complete reproduction instructions.
Moreover, we augment the code and data by also providing the paper's figures, and, for those generated by code, also the code to re-create them.

\subsection{Step I: Download}
The complete material needed to reproduce our work can be obtained from:
\begin{quote}
    \texttt{\url{https://drive.proton.me/urls/AXKHYN7KKM\#5LikpW13ajpH}}
\end{quote}

\subsection{Step II: Unpack}
The material is given as a \texttt{zip} archive.
By unpacking it, the following folders are extracted:
\begin{enumerate}
    \item \texttt{code}, that includes all the code we use in this work.
    \item \texttt{data}, that contains all the data we use in this work.
    \item \texttt{figures}, that contains all the figures we use in this work.
    We note that, by executing our code, the specific figures created by the supplied code will be re-produced.
    So, this will result in overwriting the figures included in the archive.
    The included figures correspond to the very same versions that are used in this paper.
    \item \texttt{requirements}, that has files detailing the names and versions of all packages needed to execute our code.
\end{enumerate}

\subsection{Step III: Install}
Install the packages required to run our code by using one of the following files:
\begin{enumerate}
    \item \texttt{requirements/requirements.txt}, a succinct list of the packages used in our code.
    These packages can be installed using \texttt{Python v3.13.5} and \texttt{pip v25.1} by running:
    \\\texttt{pip install -r requirements.txt}

    \item \texttt{requirements/pip\_env.txt}, a detailed list of all python packages installed on our computer when we ran our code, including packages unused by our code.
    The packages contained in this list can be installed in a similar manner for \texttt{requirements.txt}.

    \item \texttt{requirements/conda\_env.yml}, a specification of the complete Anaconda environment that we used to run our code, including packages unused in our code.
    This environment can be installed using \texttt{Anaconda v2025.06-0} by running:
    \\\texttt{conda env create -f /path/to/conda\_env.yml}
\end{enumerate}

\subsection{Step IV: Run}
To run our code, execute our main code file:
\begin{quote}
    \texttt{code/code.ipynb}
\end{quote}
We note that we use a Jupyter notebook so that all outputs created by our code and used in this submission can be viewed by simply opening the file, without running it.

\section{Glossary}
\label[appendix]{sec:Glossary}
We provide a summary of all the notations and acronyms used throughout this work.
\setglossarystyle{alttree}\glssetwidest{AAAA}
\printnoidxglossary[type={symbols}]
\printnoidxglossary[type={acronym}]
\end{document}